\documentclass[aps,pre,onecolumn,nofootinbib,floatfix,longbibliography,superscriptaddress]{revtex4-2}

\usepackage{bm}

\setcitestyle{numbers,sort&compress}

\usepackage[utf8]{inputenc}
\usepackage[english]{babel}
\usepackage[T1]{fontenc}
\usepackage{listings}
\usepackage{hyperref}
\usepackage{thm-restate}
\usepackage{float}

\usepackage{amssymb}
\usepackage{amsthm}
\usepackage{thmtools}
\usepackage{bbm}
\usepackage{qcircuit}
\usepackage{adjustbox}

\expandafter\let\csname equation*\endcsname\relax

\expandafter\let\csname endequation*\endcsname\relax

\newtheorem{thm}{Theorem}

\newtheorem{lem}[thm]{Lemma}
\newtheorem{prop}[thm]{Proposition}
\newtheorem{cor}[thm]{Corollary}
\newtheorem{obs}[thm]{Observation}
\newtheorem{defn}{Definition}

\usepackage[resetlabels]{multibib}
\usepackage{physics}
\usepackage{xcolor}
\usepackage{graphicx}
\usepackage{wrapfig}
\usepackage{qcircuit}

\usepackage{tikz}
\usepackage{tikz-cd}
\usetikzlibrary{positioning}
\tikzset{
  symbol/.style={
    draw=none,
    every to/.append style={
      edge node={node [sloped, allow upside down, auto=false]{$#1$}}}
  }
}

\def\Tr{\text{Tr}}

\theoremstyle{plain}

\begin{document}

\preprint{}

\title[A Study in Thermal: Advantage framework for resource engines]{A Study in Thermal: \\ Advantage framework for resource engines}

\author{Jakub Czartowski}
\email{jakub.czartowski@tcd.ie}
\affiliation{School of Physical and Mathematical Sciences, Nanyang Technological University, 21 Nanyang Link, 637361 Singapore, Republic of Singapore}

\author{Rafa{\l} Bistro{\'n}}
\affiliation{Doctoral School of Exact and Natural Sciences, Jagiellonian University, ul. Łojasiewicza 11, 30-348 Kraków, Poland}
\affiliation{Faculty of Physics, Astronomy and Applied Computer Science, Jagiellonian University, ul. Łojasiewicza 11, 30-348 Kraków, Poland}

\date{August 19, 2025}

\begin{abstract}
    Thermal engines have been one of the principal topics of thermodynamics ever since its beginning. 
    Today, in the era of the second quantum revolution, although the thermal processes in constant temperatures are relatively well understood in the language of resource theories, a framework to describe thermal engines is still in its infancy. In this work we formalize the resource theory of engines, initially put forward in~\cite{KAMIL}, and define such quantities as engine efficiency. Then, we 
    turn to a detailed study of thermal engines based on free operations arising from the resource theory of athermality under different restrictions: thermal operation, semilocal thermal operations and local thermal operation with classical communication.   
    In order to provide analytic lower bounds for thermal engine operation
    we construct tree-states -- free states, which can be obtained from Gibbs state in simple protocol consisting solely of two-level operations.
    Furthermore, we derive a full description of the engine based on semilocal thermal operations by describing free states, faithful monotones and catalytic advantages.
\end{abstract}

\maketitle

\section{Introduction}

    Ever since their introduction, heat engines have been the workhorse of Industrial Revolution, converting imbalance between the internal hot and ambient cold temperatures into an indispensable resource -- work. Heat engines have been subject of in-depth study, using the tools of classical thermodynamics, with one of the cornerstone results being given by the idea of Carnot engine, which operating between given cold and hot temperatures provides an optimal ratio between the heat going in, and the amount of work going out~\cite{carnot1824reflexions}. Thus, all the way back in 1824 Carnot set the ultimate limits to what a macro-scale heat engine can provide us with. In modern times of quantum mechanics, it has been realized that even few-level quantum systems can already be seen as models of heat engines~\cite{Geusic1959threelevel, Scovil1959threelevel}, thus justifying the advent of the novel field of quantum thermodynamics~\cite{RAlicki_1979}. Since then, a significant amount of research has been devoted to understanding the way in which quantum systems can act as heat baths, working media and heat engines themselves~\cite{Tonner2005,YingNg2017, Tajima2017, Mitchison2019, Bera2021,bera2021qheatcarnot}
    
    In parallel, in modern times, we are on the verge of the second quantum revolution, fuelled by the advances in the field of quantum information theory. Its influence is believed to provide enhancements in fields such as computation, secure communication, metrology and more~\cite{larose2024brief, Portmann2022security, Giovannetti2011}. In order to attain many of the aforementioned advantages, the need for some form of quantum non-locality is recognised, the most prevalent and best studied being the phenomenon of quantum entanglement~\cite{H4entanglement}. In recognition of its value, it has been encapsulated in a form of resource theory of entanglement, with entanglement as the resource at its very core~\cite{Bennett1996, Vedral1997}. The family of resource theories extended over time to such resources as coherence~\cite{Streltsov2017}, magic~\cite{Veitch2014}, athermality~\cite{Brandao2013} and more~\cite{Gour2019, gour2024resources}. In a recent development, a concept of resource engines has also been put forward. Resource engines combine two resource theories of similar nature, eg. athermality with respect to two different reservoirs or coherence with respect to two different bases, are used to drive the system initiated in a free state as far as possible from the original free states of both theories~\cite{KAMIL}. In another important recent development, a framework for analysing conversion of states out of thermal equilibrium into entangled states using solely thermodynamic means has been put forward~\cite{deoliveira2024entanglement}, with  complementary results on heat-based witnesses of entanglement proposed in~\cite{deOliveira2025heatentwitness}.

    We revisit the idea of resource engines from the viewpoint of operational advantages coming from the possibility of driving the states away from free states, or ``resourceless equilibria''. In particular, the operational tasks in question may, although not necessarily, be related to monotones of the underlying resource monotones. This presents a possibility of constructing \textbf{resource theory of engines}, where monotones are connected to the corresponding advantage quantifiers. This aligns with an emerging direction to fuse different resource theories together~\cite{son2024robust}. We then focus on simple finite-dimensional engines with operations constrained by thermal restrictions with varying degrees of control over involved systems and baths. Our examples demonstrate engines' ability to provide an advantage in two tasks related directly to thermodynamics -- \textit{cooling} and \textit{heating} -- and problem of obtaining resource crucial for quantum-informational tasks, \textit{entanglement} from athermality generated solely from the action of the engine itself.

    This work is organised as follows. In Section~\ref{sec:prelims} we provide the necessary preliminaries, where we discuss basic notions of the resource theoretic framework, different free operation sets used in the context of resource theory of athermality and basic aspects of bipartite entanglement. In Section~\ref{sec:op_ad} we provide a formalisation of the notion of resource engines and introduce measures of advantage that can be applied to any operational task with a quantitative measure of performance. Section~\ref{sec:ad_therm_eng} considers an application of thermal engines to three specific operational tasks -- cooling, heating and entanglement generation. We first present numerical and analytical analyses of engines operating under different sets of free operations, ordered by the degree of control required to implement them. To support this analysis, we introduce the notion of tree states, which provide analytical interior approximations to the sets of free states for engines based on elementary thermal operations and their restrictions. We conclude with a complete characterization of thermal engines governed by semilocal thermal operations. We close the work in Section~\ref{sec:conc} by providing a discussion of the presented results and outlook for potential future research directions. 
    
\section{Preliminaries}
\label{sec:prelims}
    In what follows, we will lay down the basic notions necessary to understand the remainder of the work, starting with a short outline of the resource-theoretic framework. Then we will present the relevant set of approaches to quantum thermodynamics, spanning Gibbs-preserving and thermal operations through semilocal thermal operations to local thermal operations and classical communication. Then we will move on to discuss the basics of entanglement theory for bipartite systems. 
    Reader acquainted with these basic notions may skip directly to section~\ref{sec:op_ad}.

    \subsection{Outline of resource theories}

    We begin the exhibition of the tools by recalling the basic definitions connected with resource theories, starting with the concept itself from a more intuitive, heuristic point of view. In day-to-day life certain objects -- like wooden sticks, for instance -- are readily available, and thus present no value; as such, sticks can be seen as \textit{free}. High quality timber, on the other hand, is hard to obtain and can be used to create useful objects -- hence, it is seen as valuable or \textit{resourceful}. Wood can also be cut at relatively low cost, given one has a handsaw at their disposal. Using it, one can turn a large piece of timber into smaller planks of different sizes, leaving some waste product behind and reducing its value; on the other hand, the sticks, even if cut down, remain free. This leads to a simplified statement that simple cutting of wood is \textit{a free operation}. This can be contrasted with detailed carving and processing of wood, which may produce furniture, instruments and art -- objects of increased value compared to the original timber. Such processing, however, requires skill and specialized equipment, thus it is not free.

    Naturally, the example of timber as described above is subject to economic laws of supply and demand. Its value, defined intersubjectively by the society, cannot be described in a rigorous manner. The situation in physics is quite different -- one can rigorously define, manipulate and quantify resources necessary for certain tasks. 
    \medskip

    Usually, formal construction of a resource theory starts from a subset $\mathcal{F}\subset\mathcal{S}$ of the full space of states, which is referred to as \textit{the set of free states}. The problem of valuation and free manipulation of resources is then handled by \textit{a monotone} $\mathcal{M}$ (or a set of monotones) and \textit{a set of free operations} $\mathcal{O} = \qty{O:\mathcal{S}\mapsto\mathcal{S}}$. Free operations should not create resource from free states:  for any state $f\in\mathcal{F}$ and operation $O\in\mathcal{O}$ we have $O(f)\in\mathcal{F}$. The monotone $\mathcal{M}:\mathcal{S}\mapsto\mathbb{R}_+$ maps the states to the non-negative real numbers, providing a measure of their values. Using only free operations one should not be able to create additional resources. This implies \textit{monotonicity of the monotone under free operations}: for any $s\in\mathcal{S}$ and $O\in\mathcal{O}$ we have $\mathcal{M}(O(s)) \leq \mathcal{M}(s)$. Further, we say that monotone $\mathcal{M}$ is \textit{faithful} whenever it vanishes only on the free states, $\mathcal{M}(f) = 0\Longleftrightarrow f\in\mathcal{F}$.

    The idea of rigorous quantification of resources and their manipulation has been particularly popular in the context of quantum theory, which has seen the advent of a whole community of quantum resource theory. Examples of quantities considered to be resources, useful for leveraging advantages presented by quantum phenomena, include coherence, negativity, non-stabilizerness or imaginarity, among others. For a full review, we refer the reader to a textbook by Gilad Gour~\cite{gour2024resources}. In this work, however, we will be concerned mainly with two resources -- entanglement, and athermality.

    \subsection{Quantum thermodynamics frameworks}

        In what follows we consider transformations of systems in contact with thermal reservoirs characterised by inverse temperature $\beta = (k_B T)^{-1}$. Before presenting specific frameworks we will utilise, including thermal operations (TOs), semilocal thermal operations (SLTOs) and local thermal operations and classical communication (LTOCC), we lay down basic notions and notations common to all of them.

        The primary system of dimension $d$, prepared in an initial state $\rho$ is characterised by its Hamiltonian $H = \sum_{i} E_i\op{E_i}$, where we assume the energy levels $E_i$ to be non-degenerate. 
    	The total system-bath initial state is assumed to be uncorrelated, $\rho\otimes\gamma_E$, with the bath, characterised by the Hamiltonian $H_E$, starting in the Gibbs state defined as
    	\begin{equation}\label{eq:gibbs_state}
    		\gamma_E = \frac{e^{-\beta H_E}}{\Tr(e^{-\beta H_E})}.
    	\end{equation}
        Hereafter subscript $E$ denotes the thermal environment. 
    	If not stated otherwise, in what follows we will be restricting our considerations of states under thermodynamic evolution to the energy-incoherent states, ie. such that they commute with the Hamiltonian of the primary system, $\comm{\rho}{H} = 0$. With these restrictions, there is a one-to-one correspondence between states $\rho = \sum_i p_i\op{E_i}$ and their population vectors $\vb{p}$. Finally, the thermal equilibrium state of the system ${\gamma}$ is defined by Eq.~\eqref{eq:gibbs_state} with $H_E$ replaced by the system Hamiltonian $H$.

        A common aspect of all frameworks under consideration is the preservation of thermal equilibrium, which is realised in the most generic manner by \textit{Gibbs-preserving operations} (GP). A map $\mathcal{E}$ is called Gibbs-preserving when Gibbs stats $\gamma$ is its fixed point, $\mathcal{E}(\gamma) = \gamma$. In particular, as we will be considering also bipartite systems, when one part of the system is in contact with a cold bath with inverse temperature $\beta_C$ and the other with a hot bath with inverse temperature $\beta_H$, the GP condition will correspond to preservation of the product of local thermal equilibria, $\mathcal{E}(\gamma_C\otimes\gamma_H) = \gamma_C\otimes\gamma_H$.
    
        \subsubsection{Thermal operations}

            In order to make the Gibbs-preserving condition operational, it is useful to consider the framework of thermal operations (TOs)~\cite{Janzing2000, Brando2015, Lostaglio2019}. It starts from the assumption that the system and the bath, prepared in an uncorrelated state $\rho\otimes\gamma_E$, undergo a joint unitary evolution as a closed system
            Thus the final state of the system with disregarded bath is given by
        	\begin{equation}
        		\mathcal{E}(\rho) = \Tr_E\qty[U\qty(\rho\otimes\gamma_E)U^\dagger].
        	\end{equation}
            Moreover, to ensure energy conservation,
            it is assumed that the unitary operation $U$ commutes with the joint Hamiltonian,
        	\begin{equation}
        		\label{E_prev}
        		\comm{U}{H\otimes\mathbbm{1}_E+\mathbbm{1}\otimes H_E} = 0,
        	\end{equation}
            so it preserves the energy of the joint system.
        	
            Given a pair of states $\rho$ and $\sigma$, the transition is possible if and only if there exists a thermal operation $\mathcal{E}$ such that $\mathcal{E}(\rho) = \sigma$ and $\mathcal{E}(\gamma) = \gamma$.
            However, as we limit ourselves to energy-incoherent states, the condition reduces to the existence of a stochastic matrix $\Lambda$ acting on the respective populations $\vb{p} = \operatorname{diag}_H(\rho),\,\vb{q} = \operatorname{diag}_H(\sigma)$ such that
        	\begin{align}
            \label{thermal1}
        		\Lambda\vb{p} = \vb{q},&&
        		\Lambda\boldsymbol{\gamma} = \boldsymbol{\gamma},
        	\end{align}
            where $\operatorname{diag}_H(\rho)$ denotes the diagonal of state $\rho$ in the Hamiltonian eigenbasis. 
        	
            In the particular case of infinite temperature, $\beta = 0$, the Gibbs state reduces to the maximally mixed state, with constant populations $\eta_i = \frac{1}{d}$, and the Gibbs-preserving condition $\Lambda \boldsymbol{\eta} = \boldsymbol{\eta}$ defines the set of noisy operations, acting on the energy-incoherent states, equivalent to the set of bistochastic matrices.
        
        \subsubsection{Semilocal thermal operations}

            Recently, there has been an effort to formulate a notion similar to thermal operations, which would be applicable to distributed systems interacting with different local thermal baths~\cite{bera2021qheatcarnot}. This approach resulted in the construction of the theory named semilocal thermal operations (SLTO), which is suitable to consider the ultimate limitations of finite-size thermal engines, as it incorporates simultaneous interactions between all systems and bath elements while implementing restrictions stemming from energy conservation and entropy-based limitations on energy transfer. 
            The exact definition of semilocal thermal operation is as follows:
    
            \begin{defn}[\cite{bera2021qheatcarnot} Definition 1]
                Let us consider a system $S$ consisting of two subsystems $S^{(A)}$, $S^{(B)}$ with Hamiltonian $H = H^{(A)}$ + $H^{(B)}$. An operation $\Lambda^{(AB)}$ is called a \emph{semilocal thermal operation} with respect to local inverse temperatures $\beta^{(A)}$ and $\beta^{(B)}$ if there exist two thermal baths $\mathcal{B}^{(A)}$, $\mathcal{B}^{(B)}$ with Hamiltonians $H^{\mathcal{B}^{(A)}}$, $H^{\mathcal{B}^{(B)}}$ such that
                \begin{equation}\label{SLTO_def}
                \Lambda^{(AB)}(\rho^{(AB)}) = \Tr_{\mathcal{B}^{(A)},\mathcal{B}^{(B)}}\left[U(\gamma^{\mathcal{B}^{(A)}} \otimes \gamma^{\mathcal{B}^{(A)}} \otimes \rho^{(AB)})U^\dagger \right]\,,
                \end{equation}
                where the global unitary matrix $U$ satisfies the following commutation relations:
                \begin{subequations}\label{SLTO_constr}
                \begin{align}
                & [U,H^{(A)} + H^{(B)} + H^{\mathcal{B}^{(A)}} + H^{\mathcal{B}^{(B)}}] = 0 ,\label{SLTO_constr1}\\
                & [U, \beta^{(A)}( H^{(A)} + H^{\mathcal{B}^{(A)}}) + \beta^{(B)}( H^{(B)} + H^{\mathcal{B}^{(B)}})] = 0 \,.\label{SLTO_constr2}
                \end{align}
                \end{subequations}
            \end{defn}
            \noindent The first constraint~\eqref{SLTO_constr1} imposes conservation of the total system-bath energy across both systems. To understand the second condition, one should consider the systems and baths in the thermodynamic limit. Then, the energy flow to or from the thermal bath corresponds to heat exchange,  
            which in this limit is equal to entropy change times the bath's temperature. Thus, the second equation~\eqref{SLTO_constr2}, describing conservation of weighted heat flow, encodes the ``second law of thermodynamics'' -- entropy conservation.
            We stress that when both thermal baths have the same temperature, semilocal thermal operations are equivalent to thermal operations on a joint system with one joint bath.
        
            The authors of~\cite{bera2021qheatcarnot} provided a strong operational bound on possible transformations of bipartite states using SLTO:
        
            \begin{thm}\label{STLO_thermo_m}[\cite{bera2021qheatcarnot}, Supplementary Information, Theorem 11]\label{SLTO_majo}
             The transition between two energy-incoherent states $\rho^{(AB)} \to \sigma^{(AB)}$ can occur under semi-local thermal operation if, and only if, the diagonal of $\rho^{(AB)}$ thermo-majorizes the diagonal of $\sigma^{(AB)}$ with respect to joint Gibbs state $\gamma^{(A)}\otimes\gamma^{(B)}$.
            \end{thm}
    
            The above is equivalent to a simple statement in terms of populations:
            The transition between two energy-incoherent states $\rho^{(AB)} \to \sigma^{(AB)}$ can occur under semi-local thermal operation if and only if there exists a stochastic matrix preserving product of Gibbs states $M(\gamma^{(A)}\otimes\gamma^{(B)}) = \gamma^{(A)}\otimes\gamma^{(B)}$, which maps the spectrum of one state into the other $M(\text{diag}(\rho^{(AB)})) = \text{diag}(\sigma^{(AB)})$.
            Therefore, if one discusses only energy-incoherent states, SLTO can be effectively treated as stochastic Gibbs-preserving matrices.
    
            Furthermore, in the recent work~\cite{bistron2024local}, 
             authors demonstrated that the set of SLTO is closed under the composition of operations, convex and (topologically) closed for finite-dimensional systems.
            Therefore, SLTO constitute an elegant resource theory akin to thermal operations, with a product of local Gibbs states (with different temperatures) as a free state and semilocal thermal operation as free operations.
            
        \subsubsection{Local Thermal Operation and Classical Communication}

            The last resource theory we are going to utilize is the recently postulated theory of local thermal operations and classical communication (LTOCC)~\cite{bistron2024local}.  The backbone of this theory is the distant laboratories paradigm: an assumption that the joint state is distributed between at least two separate parties, which can perform measurements, apply thermal operations on corresponding subsystems and communicate classically. Thus, this framework combines the properties of LOCC theory of entanglement with thermal constraints.
            Since the subsystems do not interact with each other, the global Gibbs state of the system is a product of local Gibbs states.
            We start exposition of LTOCC by introducing the building block of the basic protocol -- one round local thermal operations with classical communication, followed up by relevant extensions.

            We restrict our considerations to bipartite scenario, with parties traditionally referred to as Alice and  Bob. In the simplest non-trivial protocol, after each of them receives their own part of joint state, Alice performs a projective measurement on her subsystem obtaining some result $k$. Then she sends information to Bob who, depending on the received information, can perform different thermal operations $\Lambda_{jl}^{(k)}$. Stacking Bob's operations conveniently in single order-3 tensor $T_{jkl}^{(B)} :=\Lambda_{jl}^{(k)}$, the joint operation can be represented as a stochastic matrix   
        	\begin{equation*}
        		M_{ij,kl} = \delta_{i,k} T_{jkl}^{(B)}~,
        	\end{equation*}
            where $k,l$ are input indices and $i,j$ output indices of Alice's and Bob's subsystems, respectively.
            Since Alice only performs a measurement, the local ensemble of Alice's states remains unchanged which corresponds to $\delta_{i,k}$. On the other hand,
            since each of Bob's conditional operations preserves Bob's Gibbs state $\gamma^{(B)}$, the tensor $T_{jkl}^{(B)}$ satisfies
            \begin{equation}
            \forall_k \sum_l T_{jkl}^{(B)} \gamma_l^{(B)}  =  \gamma_j^{(B)}~.
            \end{equation}
            Finally, while Bob applies conditional thermal operation $T_{jkl}^{(B)}$, Alice can perform thermal post-processing on her system $\Lambda_{i,k}^{(A)}$. 
            Combination of these actions defines one round  \textit{Local Thermal Operations and Classical Communication} (LTOCC), described by a stochastic matrix
        	\begin{equation}
        		M_{ij,kl} = \Lambda_{i,k}^{(A)} T_{jkl}^{(B)}~,
        	\end{equation}
            without summation involved.
            
            \begin{figure}[t]
            \centering
            \begin{minipage}[h]{\textwidth}
            \centering
            \scalebox{1}{
            \Qcircuit @C=1.0em @R=1.0em @!R { \\
            	\lstick{{m}_{1} :  }&  \ar @{<=} [1,0] & \control \cw \cwx[2]  & \barrier[0em]{2} &  & \ar @{<=} [2,0] & \control \cw \cwx[1] & &\\
            	\lstick{A :  }& \meter & \qw & \gate{\Lambda^{(A,1)}} & \qw & \qw & \gate{T^{(A,2)}} & \qw & \qw \\
            	\lstick{B :  }& \qw & \gate{T^{(B,1)}} & \qw & \qw & \meter & \qw & \gate{\Lambda^{(B,2)}} & \qw \\  \\}}
            \end{minipage}
            \caption{
            An exemplary realization of two rounds of LTOCC protocol in circuit representation, with $T^{(B,1)}$ and $T^{(A,2)}$ representing the thermal operations conditioned on the memory line $m_1$ and $\Lambda^{(A,1)}$ and $\Lambda^{(B,2)}$ representing thermal post-processing.}
            
            \label{fig:ciruits}
            \end{figure}
        
            The above protocol can be extended to multiple rounds, by composing channels $M$, or by additionally introducing memory, which allows to condition operations in one round based on measurement outcomes from previous rounds.
        
            Furthermore, one can introduce shared randomness to the frameworks by adding some random probability $\lambda$ conditioning LTOCC operations:
            \begin{equation*}
            M = \sum_\alpha \lambda_\alpha M^{\alpha}~.
            \end{equation*}
            As presented in~\cite{bistron2024local} the introduction of shared randomness effectively makes the set of LTOCC operations convex. Unless stated otherwise, we will use access to shared randomness by default.
            \medskip
            
            Before finishing this subsection we present the simplest protocol with nontrivial memory usage, which can be performed in two rounds -- parallel LTOCC. 
        	Assume that after Alice's measurement, Bob performs identity operations but preserves obtained information. In the next round, Bob performs a measurement of his state and sends the information to Alice. Only then both parties perform conditioned thermal operations based on collaborator's measurement outcome. Because measurement in the energy eigenbasis does not alter the ensemble of received states, we can describe this joint operation via a stochastic matrix
        	\begin{equation}
        		\label{one_round}
        		M_{ij,kl} = T_{ikl}^{(A)} T_{jkl}^{(B)}\,,
        	\end{equation}
        	which constitutes \textit{parallel} Local Thermal Operations and Classical communication. Because thermal tensors $T$ are collections of thermal operations, each of them preserving the Gibbs state individually, we have 
        	\begin{equation}
            \label{LTLOCC_req}
        		\forall_l\sum_{k} T_{ikl}^{(A)} {\gamma}_k^{(A)} = {\gamma}_i^{(A)} \;,~~ 
        		\forall_k\sum_{l} T_{jkl}^{(B)} {\gamma}_l^{(B)} = {\gamma}_j^{(B)} \;.
        	\end{equation}
            Moreover, since composition of two thermal operations is a thermal operation, potential (unconditional) thermal pre- and post- processing can be included in thermal tensors $T$.
            One can also demand additional symmetric constraints defining \textit{Symmetric} Local Thermal Operations and Classical Communication ($\mathcal{S}$LTOCC)
        	\begin{equation}
            \label{LTLOCC_req2}
        		\forall_k\sum_{l} T_{ikl}^{(A)} {\gamma}_l^{(B)} = {\gamma}_i^{(A)} \;,~~ 
        		\forall_l\sum_{k} T_{jkl}^{(A)} {\gamma}_k^{(A)} = {\gamma}_j^{(A)} \;,
        	\end{equation}
            which states that Bob's output marginal must be a Gibbs state if only Alice's input was a Gibbs state and vice versa.

    \subsection{Bipartite entanglement}

        A bipartite quantum state $\rho\in\mathcal{D}(\mathcal{H}_A\otimes\mathcal{H}_B)$ is entangled if it cannot be written as a convex combination of products of local states, $\rho_{AB}\neq \sum_i p_i \rho_{A,i}\otimes\rho_{B,i}$. For pure states $\ket{\psi_{AB}}\in\mathcal{H}_A\otimes\mathcal{H}_B$ the statement reduces to saying that the state is not of the product form, $\ket{\psi_{AB}}\neq\ket{\psi_A}\otimes\ket{\psi_B}$. This can be seen as a consequence of the fact that pure states are extremal points of the set of mixed states, thus are themselves non-decomposable. Hence stating that a given state does not belong to the product sub-manifold is sufficient. In consequence, Schmidt decomposition $\ket{\psi_{AB}} = \sum_{i=1}^d \sqrt{\lambda_i}\ket{\psi_{A,i}}\otimes\ket{\psi_{B,i}}$ and the corresponding Schmidt vector provide a complete description of the entanglement structure for a given pure state -- a fact that has been leveraged in considerations of transformability between pure entangled states in terms of majorisation relations between the respective Schmidt vectors~\cite{H4entanglement}.

        For mixed states the problem of deciding whether a given state is entangled or not is relatively complicated. One possible way is via application of a positive, but not completely positive, trace-preserving map to one of the subsystems and examination whether the output is still a valid quantum state. This basic intuition stands behind the positive partial transpose (PPT) criterion, which states that a state $\rho$ is separable only if its partial transpose is positive, $\rho^{\Gamma_A}\geq 0$. This provides a basis for a measure of entanglement called negativity and defined as
        \begin{equation}
            \mathcal{N}(\rho) = \frac{\norm*{\rho^{\Gamma_A}}_1 - 1}{2}.
        \end{equation}
        Due to the statement above, we know that $\mathcal{N}(\rho) = 0$ whenever a state is not entangled. However, it is known that negativity is a faithful monotone only for systems of dimension $2\times2$ and $2\times 3$~\cite{peres1995quantum, horodecki1996separability}; for any larger systems there exist PPT entangled states.
        In particular, for two qubits we are guaranteed to have at most one negative eigenvalue $\lambda_{min}$ for the partially transposed state $\rho^{\Gamma_A}$, and thus negativity reduces to
        \begin{equation}
            \mathcal{N}(\rho) = \max(-\lambda_{min}, 0).
        \end{equation}
        Utilizing its faithfulness and simplicity, in what follows we will use negativity to quantify the degree of entanglement for two-qubit systems.

\section{Operational advantage for resource engines}
\label{sec:op_ad}

Recently a compelling concept of resource engines has been introduced, combining two (or more) input resource theories to drive otherwise free states out of the free set of either of the theories~\cite{KAMIL}. 
At a conceptual level, resource engines mirror the operation of conventional thermal engines: they proceed in strokes that sequentially couple the system to baths at different temperatures, or, more generally, switch between distinct resource theories. This sequential structure enables operations that are unattainable under a single-bath assumption. Unlike traditional stroke-based engines, whose primary objective is the production of work, resource engines are designed to drive states that are free in one theory far outside the free sets of the constituent theories. In doing so, they effectively enlarge both the set of accessible states and the set of achievable operations, without introducing any capabilities beyond those already present in the individual resource theories composing the engine.

\medskip

More specifically, let us consider two resource theories, characterized by triples $\mathcal{R}_1 = (\mathcal{F}_1,\mathcal{O}_1,\mathcal{M}_1)$ and $\mathcal{R}_2 = (\mathcal{F}_2,\mathcal{O}_2,\mathcal{M}_2)$, providing an analogue of operations allowed in equilibrium with a ``cold'' and ``hot'' bath from a classical stroke-model thermal engine. Following the idea introduced in~\cite{KAMIL}, we construct engine theories $\mathfrak{E}(\mathcal{R}_1,\,\mathcal{R}_2)\equiv \mathfrak{E}$, characterized by a pair $(\mathfrak{F},\mathfrak{O})$, where the set of free states contains sets of free states of the starting theories, $\mathfrak{F}\supset \mathcal{F}_1\cup\mathcal{F}_2$; similarly, the set of free operations is assumed to contain sets of free operations from the underlying theories $\mathfrak{O}\supset\mathcal{O}_1\cup\mathcal{O}_2$ and, additionally, is assumed to be closed under composition.

We begin by putting forward the following property of the set of free operations in any engine theory.
\begin{prop}
\label{prop_combine}
    Given an engine theory $\mathfrak{E}\supset\mathcal{R}_1\cup\mathcal{R}_2$ the set of free operations is given in general by
    \begin{equation}
    \label{eq:1st_def_free_O}
        \mathfrak{O} = \qty{\prod_{i=1}^k \hat{O}_i,\,\hat{O}_i\in\mathcal{O}_1\cup\mathcal{O}_2, k \in \mathbb{N} \cup \{\infty\}}.
    \end{equation}
    Furthermore, if both $\mathcal{O}_1$ and $\mathcal{O}_2$ are closed and have a convex structure, to preserve the convexity of the theories one can define it in terms of extreme points, as the smallest closed convex set containing the composition of extremal free operations: 
    \begin{equation}\label{eq:2nd_def_free_O}
        \mathfrak{O} = \overline{\operatorname{conv}\qty(\qty{\prod_{i=1}^k \hat{O}_i,\,\hat{O}_i\in\operatorname{ext}\qty[\mathcal{O}_{(i\operatorname{mod}2)+1}],k \in \mathbb{N} \cup \{\infty\}})}.
    \end{equation}
\end{prop}
\begin{proof} 
    Equation~\eqref{eq:1st_def_free_O} is trivially satisfied by demanding that $\mathfrak{O}$ is closed under composition.
    
    To show that~\eqref{eq:2nd_def_free_O} holds, let us first note that the composition of two consecutive free operations from one resource theory results in another free operation within this theory, so we may only consider alternating sequences of operations.
    
    If we restrict ourselves to a finite number of strokes $k < \infty$ then any term in a product in~\eqref{eq:1st_def_free_O} can be represented as a finite convex composition of extremal free operations, thus their product can be finitely decomposed as well. Moreover, all extremal points of $\mathfrak{O}$ must be products of extremal operations, which justify the proposition.
    The limit $k\to \infty$ follows from the fact that all extremal points of $\mathfrak{O}$ must be once again, products of extremal operations. Hence, all other elements of $\mathfrak{O}$ are convex combinations of the former. 
\end{proof}

Similarly, one can easily formulate similar statement for the set of free states.

\begin{prop}
    \label{prop_combine_states}
    Given an engine theory $\mathfrak{E}\supset\mathcal{R}_1\cup\mathcal{R}_2$ the set of free states is given by
    \begin{subequations}
        \begin{align}
            \mathfrak{F} & = \qty{\hat{O} f,\,\hat{O}\in\mathfrak{O},\,f\in\mathcal{F}_1\cup\mathcal{F}_2} \\
            & = \operatorname{conv}\qty(\qty{\hat{O} f,\,\hat{O}\in\operatorname{ext}(\mathfrak{O}),\,f\in\operatorname{ext}(\mathcal{F}_1)\cup\operatorname{ext}(\mathcal{F}_2)}) \,,\label{eq:free_states_engine} 
        \end{align}
    \end{subequations}
    with~\eqref{eq:free_states_engine} holding if all free sets $\mathcal{F}_1,\,\mathcal{F}_2,\,\mathfrak{O}$ are convex, with elements of $\mathfrak{O}$ linear,  and where $\operatorname{ext}$ denotes the set of extremal points.
\end{prop}

A joint interpretation of Propositions~\ref{prop_combine} and~\ref{prop_combine_states} is as follows. The operation of a resource engine consists of an arbitrary sequence of alternating strokes that combine free operations from the constituent resource theories, thereby substantially enlarging the set of actions available to the engine. If the original resource theories admit probabilistic mixtures of operations, this property is naturally inherited.
Moreover, the set of free states is not restricted to the ``original’’ free states of the individual theories, but also includes all states obtained by applying free operations to free states of any constituent theory. As a consequence, free operations can drive these states, at no additional cost, into states that would otherwise be considered resourceful.
It is relatively simple to present conditions under which an engine theory $\mathfrak{E}$ reduces to one of the underlying theories.

\begin{prop}\label{prop_4}
    Engine theory $\mathfrak{E}\supset\mathcal{R}_1\cup\mathcal{R}_2$ is equivalent to one of the underlying resource theories $\mathcal{R}_1$ if and only if $\mathcal{F}_2\subseteq \mathcal{F}_1,\,\mathcal{O}_2\subseteq\mathcal{O}_1$.
\end{prop}
\begin{proof}
    First, let us assume that $\mathcal{F}_2\subseteq \mathcal{F}_1,\,\mathcal{O}_2\subseteq\mathcal{O}_1$. From the second inclusion, we immediately obtain $\mathcal{O}_1 = \mathfrak{O}$. Furthermore, from the first inclusion gives $\mathcal{F}_1\cup\mathcal{F}_2 = \mathcal{F}_1$, which is an invariant set under operation from $\mathcal{O}_1$. Thus, the implication in the first direction is fulfilled.

    To prove the implication in the opposite direction, without loss of generality, let us assume that $\mathfrak{O} = \mathcal{O}_1$. Since $\mathcal{O}_2 \subseteq \mathfrak{O}$, which can be seen by setting $k = 1$ in Proposition~\ref{prop_combine} and taking any $O_1 \in \mathcal{O}_2$, we obtain $\mathcal{O}_2\subseteq\mathcal{O}_1$. Similarly, one can show that $\mathcal{F}_2 \subseteq \mathfrak{F} = \mathcal{F}_1$ which ends the proof.
\end{proof}

Heuristically, engine's main task is to provide its user with an advantage at a given operational task. Thus, in what follows, we put forward possible measures of operational advantage obtained by using the imbalance between the input resource theories. To avoid trivialization, we assume that the sets of free states associated with the constituent resource theories are mutually non-inclusive. This assumption rules out the applicability of Proposition~\ref{prop_4}. Furhteremore, we require that the monotones of the different theories can distinguish their free states, in the sense that a monotone from one resource theory takes a nonzero value on at least one free state of another resource theory.
Two quantities that are relatively easily understood from intuition are \textit{relative advantage} $\mathfrak{A}_r$ and \textit{free advantage} $\mathfrak{A}_f$, defined as
\begin{subequations}
\label{amplifications}
    \begin{alignat}{3}
        &\mathfrak{A}_r(s;\mathcal{M}_i) =\,\max_{\hat{O}\in\mathfrak{O}}\;\frac{\mathcal{M}_i(\hat{O} s) - \mathcal{M}_i(s)}{\mathcal{M}_i(s)}\,,  \\
        &\mathfrak{A}_f(s;\mathcal{M}_i) =\,  
        \max_{\hat{O}\in\mathfrak{O}}
        \qty(\min_{f\in\mathcal{F}_j,\mathcal{M}_i(f)>0}
        \frac{\mathcal{M}_i(\hat{O} s) - \mathcal{M}_i(f)}{\mathcal{M}_i(f)})  =  \max_{\hat{O}\in\mathfrak{O}}\; \frac{\mathcal{M}_i(\hat{O} s) - \mathcal{M}_i(f_*)}{\mathcal{M}_i(f_*)} \,.\label{eq:free_advantage} 
    \end{alignat}
\end{subequations}

Relative advantage $\mathfrak{A}_r$ shows how much advantage, as measured by the monotone $\mathcal{M}_i$ from $i^{\text{th}}$ resource theory, has been generated by the engine in comparison with the original capability of a given state~$s$. It is worth noting that if the set of engine free states $\mathfrak{F}$ is compact, there will be a free state for which this quantity goes to zero, $\exists \tilde{f}\in \mathfrak{F}:\mathfrak{A}_r(\tilde{f};\mathcal{M}_i) = 0$. It can be easily seen by considering a free state $\tilde{f}$ which maximises $\mathcal{M}_i(\tilde{f})$, so $\mathcal{M}_i(\hat{O}\tilde{f})$ cannot increase as free operations map free states into free states.  
In particular, if the monotone and the set of free states $\mathfrak{F}$ are additionally convex, $\mathcal{M}_i$ will be maximised on the boundaries of $\mathfrak{F}$.

Free advantage $\mathfrak{A}_f$, on the other hand, shows the largest advantage in relation to all free states of engine's components, except $\mathcal{F}_i$, for which $\mathcal{M}_i$ vanishes trivially. Most importantly, since the values of $\mathcal{M}_i(\hat{O} s)$ and $\mathcal{M}_i(f)$ are independent, one can uncouple maximization over engine's operations $\mathfrak{O}$ and minimization over components' free states $\mathcal{F}_j$. In the minimization, the goal is to simultaneously lower the value of nominator and increase the value of demonstrator both of which implies maximization of $\mathcal{M}_i(f)$. Since for each monotone $\mathcal{M}_i$ there is a unique state $f_*^i\in\bigcup_{j\neq i} \mathcal{F}_j$ for which $\mathcal{M}_i(f_*^i) = \max_{f\in\mathcal{F}_j,j\neq i} \mathcal{M}_i(f)$, it simplifies the calculation. Since for any state $s$, monotone $\mathcal{M}_i$ and fixed free operation $\hat{O}_1 \in \mathfrak{O}$  one has $\max_{\hat{O} \in \mathfrak{O}} \mathcal{M}_i(\hat{O}\hat{O}_1 s )\leq \max_{\hat{O} \in \mathfrak{O}} \mathcal{M}_i(\hat{O} s )$, we can easily see that free advantage is monotonic under global free operations.

Both $\mathfrak{A}_f$ and $\mathfrak{A}_r$ are suitable to measure the advantage generated by a resource engine for any state $s$. However, for completeness of the framework we introduce a measure that differentiates between free states $\mathfrak{F}$ and the remainder of the state space constituting the states resourceful from the perspective of the resource engine. We put forward a measure referred to as \textit{engine efficiency} $\mathfrak{Ef}$, which compares the best free advantage attainable for any state from the free set $\mathfrak{F}$ and the free advantage of the input state $s$,
\begin{equation}\label{eff_def}
    \mathfrak{Ef}(s;\qty{\mathcal{M}_i}) = \max\qty(\max_{i\in\{1,2\}} \frac{\mathfrak{A}_{f}(s;\mathcal{M}_i)}{\max_{\tilde{f}\in\mathfrak{F}} \mathfrak{A}_{f}(\tilde{f};\mathcal{M}_i)}-1,0).
\end{equation}
Construction of engine efficiency completes the description of engine theory as a resource theory, by providing monotone. Indeed, the monotonicity of free advantage implies monotonicity of engine efficiency. Furthermore efficiency
is guaranteed to be zero for all $\tilde{f}\in\mathfrak{F}$ by construction. 

Finally, let us note that both measures of advantage can be extended by using an arbitrary function of benefit $\mathcal{B}$ instead of monotones $\mathcal{M}_i$, to evaluate advantage brought by the engine in tasks that may not be directly related to the underlying resource theories or their monotones, thus giving raise to efficiencies $\mathfrak{Ef}(s;\mathcal{B})$ with the benefit function $\mathcal{B}$ common to all the constituent theories of the engine, with its direct counterparts for advantages $\mathfrak{A}_r$ and $\mathfrak{A}_f$. 

In the remainder of this work, we focus primarily on engines' free states $\mathfrak{F}$. Therefore, the above-derived quantities will be used mostly for free states to characterize how well the engine works without ``extra fuel''. As a consequence, efficiency $\mathfrak{Ef}$ vanishes trivially. On the other hand free advantage $\mathfrak{A}_f$ is a suitable measure to characterize the best states within $\mathfrak{F}$ with respect to the selected operational tasks.

\section{Advantages from thermal engines}
\label{sec:ad_therm_eng}

In what follows, we focus on three operational tasks enabled by thermal engines: cooling, heating, and entanglement generation. The first two tasks arise naturally in the context of quantum thermodynamics \cite{Palao2001cooling, Levy2012QRefrigerators, Kosloff2014} and have been partially investigated from the resource-theoretic perspective in~\cite{KAMIL}. Entanglement generation from athermality has been a subject of investigation from the point of view of stead-state engineering \cite{BohrBrask2015, FrancescoDiotallevi2024, dols2025steady0state} and, by contrast, has only recently been introduced within the resource‑theoretic framework for a single thermal bath~\cite{deoliveira2024entanglement}. Nevertheless, this task is closely aligned with the conventional intuition behind heat engines, which convert a temperature imbalance into an output resource that is qualitatively distinct from the input, such as work or electricity.

Thermal engines typically operate between two thermal reservoirs in equilibrium at different temperatures, $T_c < T_h$, referred to as the cold and hot baths, respectively. Unless stated otherwise, our working medium consists of two qubits with identical Hamiltonians, having energy levels $E_0 = 0$ and $E_1 = 1$, and eigenstates given by the computational basis. The corresponding Gibbs states of a single qubit for the cold and hot baths are therefore fully characterized by their populations,
\begin{align}
    \boldsymbol{\gamma} = \mqty(\gamma,&1-\gamma) = \mqty(\frac{1}{1+e^{-\beta_c}},&\frac{e^{-\beta_c}}{1+e^{-\beta_c}}), \qquad
    \boldsymbol{\Gamma} = \mqty(\Gamma,&1-\Gamma) = \mqty(\frac{1}{1+e^{-\beta_h}},&\frac{e^{-\beta_h}}{1+e^{-\beta_h}}),
\end{align}
where $\beta_x = 1/T_x$. Since the composite system consists of two identical subsystems, we denote the joint populations by $p_{ij}$, where indices $i,j$ label the energy levels of the first and second qubit, respectively.

As a reference point, let us consider the best cooling, heating and entanglement achievable when our operations are restricted to local thermalisation with ambient temperature, which can be seen as a Markovian thermal process as described in~\cite{lostaglio2022continuous}. In what follows cooling will be quantified by ground state population $\mathcal{P}_G$, and in the case of local thermalization the maximum is achieved for
\begin{equation}
    \mathcal{P}_G\qty(\boldsymbol{\gamma}^{\otimes 2}) = \gamma^2.
\end{equation}
Similarly, heating will be quantified by maximal achievable population of the maximally excited level $\mathcal{P}_E$, which in the present case is given by
\begin{equation}
    \mathcal{P}_E\qty(\boldsymbol{\Gamma}^{\otimes 2}) = (1-\Gamma)^2.
\end{equation}
Finally, for entanglement we will consider a situation where in the final step of the protocol both qubits are brought into contact with a common thermal bath, thus allowing to entangle them using unitary operation $U$ restricted to the energy-degenerate subspace  $\operatorname{span}(\ket{01},\ket{10})$. Following results from~\cite{deoliveira2024entanglement}, maximal negativity achievable through thermal operations for a diagonal state $\rho$ with population distribution $\vb{p}$ is given by
\begin{equation}
    \mathcal{N}_{max}(\vb{p}) = \max_U \mathcal{N}(U\rho U^\dagger) = \frac{1}{2}\max\qty(\sqrt{(p_{00}-p_{11})^2 + (p_{01}-p_{10})^2} - p_{00} - p_{11}, 0),
\end{equation}
where the optimal unitary rotation $U$ turns out to be Hadamard gate.
In particular, $\mathcal{N}_{max} = 0$ whenever $4p_1 p_4 \geq (p_2-p_3)^2$. From this we find that a state $\boldsymbol{\gamma}\otimes\boldsymbol{\Gamma}$ can be entangled by the means described above if and only if
$4 \Gamma^{2} \gamma^{2} - 4 \Gamma^{2} \gamma - \Gamma^{2} - 4 \Gamma \gamma^{2} + 6 \Gamma \gamma - \gamma^{2}>0$. This translates to a simple relation between exponentiated temperatures~\cite{deoliveira2024entanglement}
\begin{equation} \label{eq:no_entanglement_sep_qub}
    \beta_c-\beta_h > \ln(3+2\sqrt{2}) \Longleftrightarrow \mathcal{N}_{max}(\boldsymbol{\gamma}\otimes\boldsymbol{\Gamma}) > 0.
\end{equation}

The above quantities -- population of the ground state, highest excited state and maximal negativity -- are presented in Fig.~\ref{fig:sep_qubits} as a function of exponentiated inverse temperatures $\exp(-\beta_c),\,\exp(-\beta_h)$.

\begin{figure}[H]
    \centering
    \includegraphics[width=\linewidth]{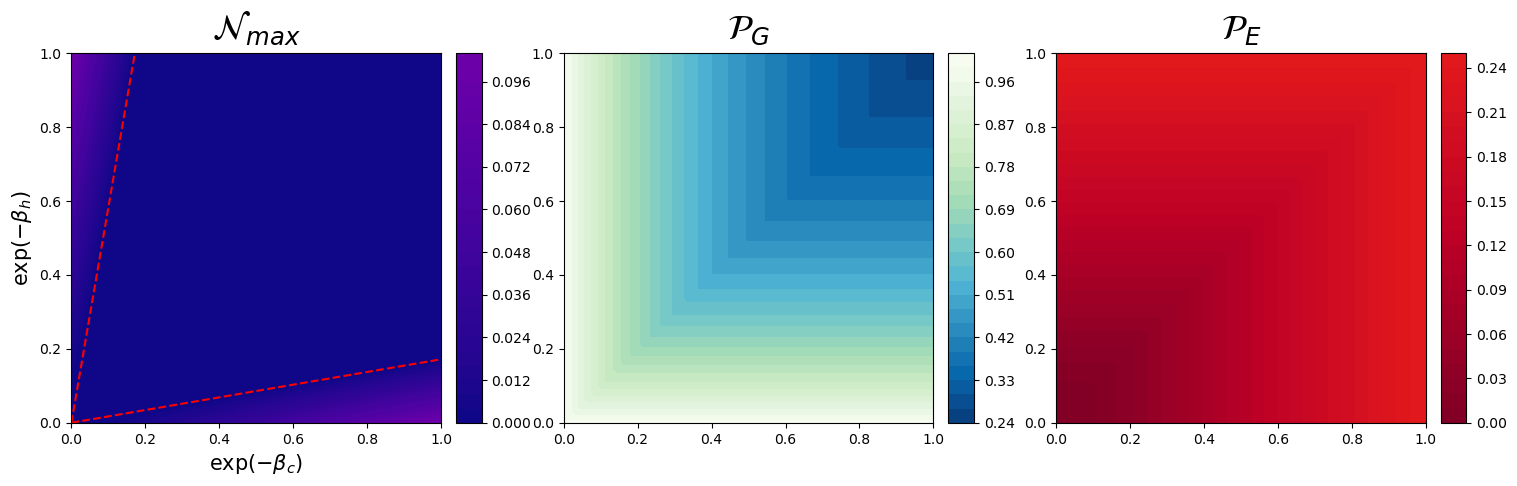}
    \caption{
    \textbf{Two equilibrated qubits:} Under assumption of no operations other than equilibration with the local baths, the space of accessible states is limited to convex combinations of $\boldsymbol{\gamma}$ and $\boldsymbol{\Gamma}$ on either subspace. Thus, optimal values of negativity $\mathcal{N}_{max}$, ground state population $\mathcal{P}_G$ and maximally excited state $\mathcal{P}_E$ are achieved for $\boldsymbol{\gamma}\otimes \boldsymbol{\Gamma}$, $\boldsymbol{\gamma}\otimes \boldsymbol{\gamma}$ and  $\boldsymbol{\Gamma}\otimes \boldsymbol{\Gamma}$, respectively. Additionally, the region between red dashed lines corresponds temperatures for which no entanglement can be generated. Note that for the sake of plots symmetry we allow $\beta_h > \beta_c$.}
    \label{fig:sep_qubits}
\end{figure}

The extremal values of $\mathcal{P}_G,\,\mathcal{P}_E$ and $\mathcal{N}_{max}$ for two equilibrated qubits, which can be seen as best values for Markovian thermal processes\footnote{Note that when $d>2$, one can do better than equilibration, as described in~\cite{lostaglio2022continuous}.}, will serve as reference points for the following sections. In particular, note that $\vb{\gamma}^{\otimes 2}$ and $\vb{\Gamma}^{\otimes 2}$ are the free states for cold and hot temperature, respectively. As such, they serve as the optimal states $f_*$ for the free advantages, defined in~\eqref{eq:free_advantage}, which we denote as $\mathfrak{A}_f(\mathbb{\gamma},\mathcal{P}_G) \equiv \mathfrak{A}_f(\mathcal{P}_G)$ and $\mathfrak{A}_f(\mathbb{\gamma},\mathcal{P}_E) \equiv \mathfrak{A}_f(\mathcal{P}_E)$, respectively. In the reminder of this work they will be frequently used to quantify the advantage consecutive engines can provide for the Gibbs state. Furthermore, note that even by allowing independent equilibration and thus treating $\vb{\gamma}\otimes\vb{\Gamma}$ as a free state of the starting theory, which would be bending the formalism, $\mathcal{N}_{max}(\vb{\gamma}\otimes\vb{\Gamma}) = 0$ for most temperatures. For these values of temperature, free advantage with respect to negativity $\mathfrak{A}_f(\cdot\,,\mathcal{N}_{max})$ is ill-defined, since no free states from the constituent resource theories might be entangled. For this reason, $\mathcal{N}_{max}$ will be considered directly.


At this point it is worth emphasising that the quantities $\mathcal{P}_G,\,\mathcal{P}_E,\,\mathcal{N}$, introduced to quantify cooling, heating, and entanglement generation, do not constitute monotones for any of the underlying resource theories, despite their clear operational significance. Nevertheless, they capture tangible performance gains provided by a resource engine operating between two athermality theories at different temperatures. This highlights an important conceptual point: in the context of resource engines, operationally meaningful benefit functions may provide a more relevant notion of performance than the monotones of the constituent resource theories themselves. Accordingly, the notions of advantage and efficiency considered in subsequent sections are defined with respect to a chosen benefit function $\mathcal{B}\in\qty{\mathcal{P}_G,\,\mathcal{P}_E,\,\mathcal{N}}$.

\medskip

In what follows, we consider thermal engines operating under progressively relaxed restrictions on the control over the subsystems, interacting cyclically with hot and cold baths. We begin with two separate qubit engines capable of performing full thermal operations, but not interactions between the qubits. Next, we move on to the case of joint bath control, where at each cycle both qubits interact with the same temperature. Within this setting, we consider one- and two-round LTOCC, elementary TOs and full 2-qubit TOs. Then allow qubits to interact with baths separately, leading to the possibility of interaction with different temperatures during a single cycle. This setting enables us to consider one- and two-round LTOCC with different local temperatures, elementary SLTOs and full two-qubit SLTOs. Finally, we shift to a radically more powerful set of operations, LTOCC with memory.

    \subsection{Two qubit engines \textit{revisited}}\label{two_qubit_simple}

    The natural stepping stone towards more complicated control models is a scenario in which two qubits can be put into contact with two separate baths without interaction between them, in order to drive them as far out of equilibrium as possible. This corresponds to utilization of local thermal operations with two different temperatures acting on two qubits separately, which we will write as $\mathcal{R}_\gamma \text{ ``='' } \mathrm{TO}_2^{\otimes 2}(\gamma)$ and analogously for $\mathcal{R}_\Gamma$. As this case has already been considered at length in~\cite{KAMIL}, let us shortly recall the results presented therein.
    
    For a single qubit, the hottest and coldest states to which one can drive the system are given by  
    $\tilde{\boldsymbol{\gamma}}  = (\tilde{\gamma},1-\tilde{\gamma})$ and $\tilde{\boldsymbol{\Gamma}} = (\tilde{\Gamma},1-\tilde{\Gamma})$ with 
    \begin{equation} \label{eq:extreme_qubit_pops}
    \tilde{\gamma} = \min\qty(\frac{(2 \gamma - 1)\Gamma}{\Gamma + \gamma - 1},1) \qq{and} 
    \tilde{\Gamma} = \max\qty(\frac{(2 \Gamma - 1)\gamma}{\Gamma + \gamma - 1},0).
    \end{equation} 
    The result is easily derived by alternative application of extreme thermal-swap~\cite{KAMIL} from both theories to the initial thermal state $\vb{\gamma}$, thus effectively swinging the state back and forth beyond both thermal equilibria using alternating temperatures. The set of free states is essentially given as interpolation between these two points, $t\tilde{\boldsymbol{\gamma}} + (1-t)\tilde{\boldsymbol{\Gamma}}$ for $0\leq t \leq 1$. 
    Note that for any state in the theory, either free or resourceful one, the alternative application of thermal swaps  will exponentially drive the state, from inside or outside of the free states' set, into either $\tilde{\boldsymbol{\Gamma}}$ or $\tilde{\boldsymbol{\gamma}}$ as described in detail in
   ~\cite{KAMIL}.

    Since for each qubit, the set of free states is a convex hull of $\tilde{\boldsymbol{\Gamma}}$ and $\tilde{\gamma}$, the free states in joint theory are products of those two sets, $\mathfrak{F} = \operatorname{conv}(\tilde{\boldsymbol{\gamma}},\tilde{\boldsymbol{\Gamma}})^{\otimes 2}$. Note that $\mathfrak{F}$ is not a convex set, since $\operatorname{conv}(\tilde{\boldsymbol{\gamma}}^{\otimes2},
    \tilde{\boldsymbol{\gamma}}\otimes \tilde{\boldsymbol{\Gamma}},
    \tilde{\boldsymbol{\Gamma}}\otimes \tilde{\boldsymbol{\gamma}},
    \tilde{\boldsymbol{\Gamma}}^{\otimes2}) \neq \operatorname{conv}(\tilde{\boldsymbol{\gamma}},\tilde{\boldsymbol{\Gamma}})^{\otimes 2}$.

    With this at hand, we may proceed to quantify the advantage brought by the engine action compared to qubits in thermal equilibrium. Maxima of $\mathcal{P}_G$ and $\mathcal{P}_E$ are obtained trivially for $\tilde{\boldsymbol{\gamma}}^{\otimes2}$ and $\tilde{\boldsymbol{\Gamma}}^{\otimes2}$, respectively, while $\mathcal{N}_{max}$ is a less straightforward quantity, which has been considered already in~\cite{deoliveira2024entanglement}, where it has been shown that for a fixed energy-incoherent input state it is optimal to utilize a Hadamard gate, restricted to the energy-degenerate subspace $\operatorname{span}(\ket{01}, \ket{10})$, to yield maximal entanglement while keeping the operation energy-preserving. Here, we assume that the setting allows for such an interaction only at the final stage, as control over any interaction between systems is costly, and thus the systems are evolved separately until the very last step. Additionally, since negativity is a convex function, one can easily show that to maximize it for any fixed state of one qubit, one has to set another qubit in an appropriate extremal state and vice versa. Thus, we may consider only the products of extremal states. The states $\tilde{\boldsymbol{\gamma}}\otimes \tilde{\boldsymbol{\gamma}}$ and $\tilde{\boldsymbol{\Gamma}}\otimes \tilde{\boldsymbol{\Gamma}}$  are maximally mixed in $\qty{\ket{01},\ket{10}}$ subspace, hence no coherences can be created. Therefore, the best free states for such protocol are $\tilde{\boldsymbol{\gamma}}\otimes \tilde{\boldsymbol{\Gamma}}$ and $\tilde{\boldsymbol{\Gamma}}\otimes \tilde{\boldsymbol{\gamma}}$, for which
    the negativity is equal 
    \begin{equation} \label{eq:max_ent_two_qub_eng}
        \begin{aligned}
            \mathcal{N}_{max}(\tilde{\boldsymbol{\gamma}}\otimes\tilde{\boldsymbol{\Gamma}})&  = \max\left[0,\frac{\gamma ^2 (-8 (\Gamma -1) \Gamma -1)+8 \gamma  (\Gamma -1) \Gamma +\gamma -\Gamma ^2+\Gamma }{(\gamma +\Gamma -1)^2}\right. \\
            + \frac{1}{2}&\left.\left(\sqrt{\frac{\gamma ^2 (16 (\Gamma -1) \Gamma +5)-2 \gamma  (\Gamma  (8 \Gamma -7)+2)+\Gamma  (5 \Gamma -4)+1}{(\gamma +\Gamma -1)^2}}-2\right)\right].
        \end{aligned}
    \end{equation}
    Although it is possible to put forward an explicit form of the curve along which $\mathcal{N}_{max}(\tilde{\boldsymbol{\gamma}}\otimes\tilde{\boldsymbol{\Gamma}}) = 0$, which should be understood as boundary between temperature pairs, for which entanglement generation using two separate qubit engines is possible, for the sake of clarity we resort to a plot in Fig.~\ref{fig:two_separate_engines}, in which the zero-curve is given by a red dashed line.

    To calculate free advantages $\mathfrak{A}_f(\mathcal{P}_G)$ and $\mathfrak{A}_f(\mathcal{P}_E)$ we first notice that the states with the maximal ground state, and excited state populations are $\tilde{\boldsymbol{\gamma}}^{\otimes2}$ and $\tilde{\boldsymbol{\Gamma}}^{\otimes2}$.  Since they consist of extremal free engine states on each qubit, the maximization over engine free operations in~\eqref{eq:free_advantage} leads exactly to them. Thus the free advantages are given by

    \begin{subequations}
        \begin{align}
            &\mathfrak{A}_f(\mathcal{P}_G) ~=~ \frac{\tilde{\gamma}^{2} - \gamma^2}{\gamma^2} ~~~~~~~~~~~~~~~\,= \frac{(1-\gamma) (\gamma -\Gamma ) (\Gamma - \gamma  (\gamma +3 \Gamma -1))}{\gamma ^2 (\gamma +\Gamma -1)^2} \,,\\
            &\mathfrak{A}_f(\mathcal{P}_E) ~=~ \frac{(1-\tilde{\Gamma})^{2} - (1-\Gamma)^2}{(1-\Gamma)^2} = \frac{(\gamma -\Gamma ) (3 \gamma +\Gamma -2)}{(\gamma +\Gamma -1)^2}\,.
        \end{align}
    \end{subequations}
    

    Note, that this specific model of a two-qubit system interacting with two heat baths is minimal when it comes to interaction between the constituents -- for heating and cooling the interaction is absent, while for entanglement it is limited to the single final step, which corresponds to a rotation within the energy-degenerate subspace necessary to turn population imbalance into coherences. As such, it will serve as a benchmark for the subsequent models, which will provide us with increasingly sophisticated control over interaction between qubits and their respective baths, as well as between the qubit systems themselves.

    \begin{figure}[H]
        \centering
        \includegraphics[width=1\linewidth]{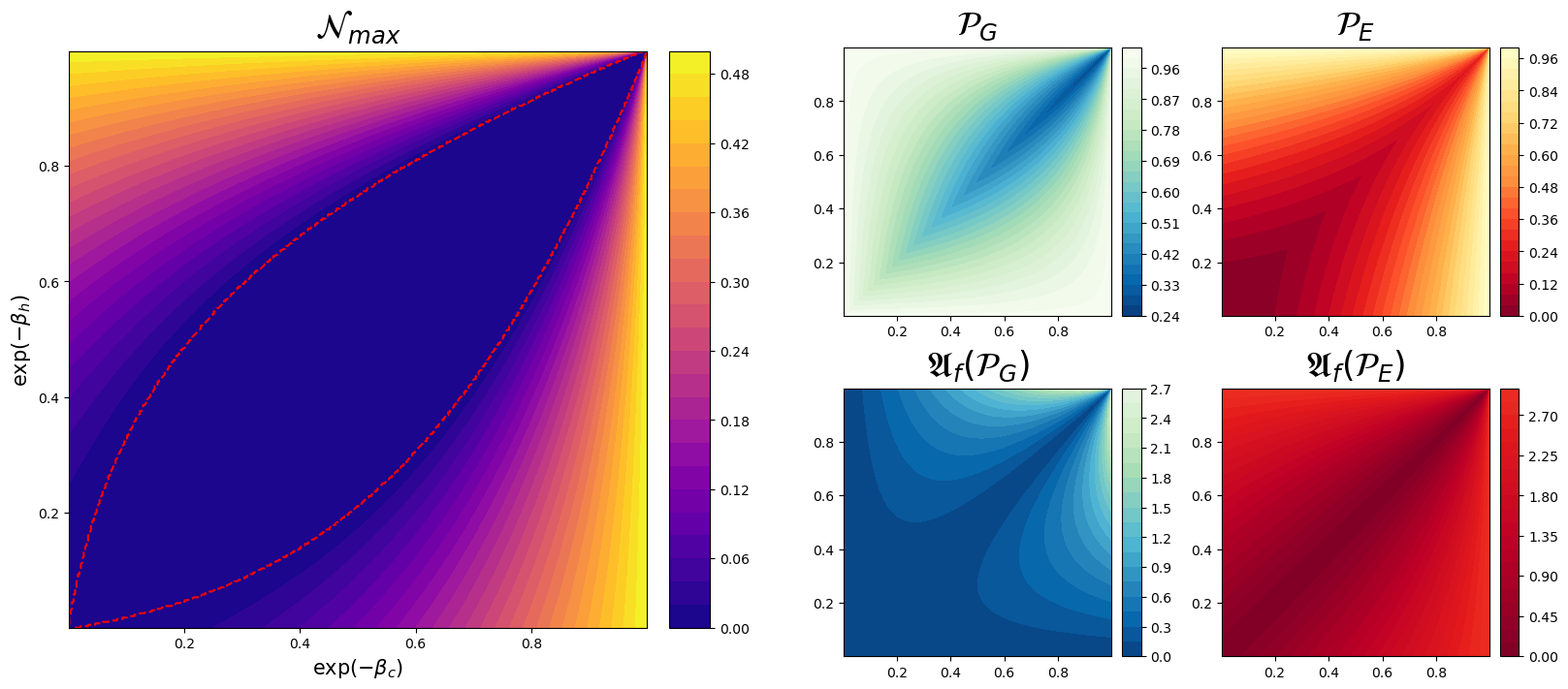}
        \caption{\textbf{Separate qubit engines:} The possibility of driving two qubits out of equilibrium shows a significant advantage over equilibrated qubits as shown in Fig.~\ref{fig:sep_qubits}. In particular, note the possibility of attaining maximal negativity $\mathcal{N}_{max}$ as $\beta_h\rightarrow 0$ (see eq.~\eqref{eq:max_ent_two_qub_eng}) and maximal excited state population $\mathcal{P}_E$ in the same limit, as prescribed by eq.~\eqref{eq:extreme_qubit_pops}. Similarly, maximal ground state population $\mathcal{P}_G$ is attainable either for $\beta_h\rightarrow0$ or $\beta_c\rightarrow\infty$. In particular, note that the no-entanglement region is limited to a tear-shaped region. Both advantages $\mathfrak{A}_f(\mathcal{P}_G)$ and $\mathfrak{A}_f(\mathcal{P}_E)$ reach their maximum for very high temperatures, while showing limited advantage in low-temperature regime. We emphasise that in this and future figures $\mathfrak{A}_f(\mathcal{P}_G)$ must tend to zero in the limit of $\beta_C, \beta_H \to 0$ since in that case any Gibbs state already has population accumulated on ground level.}
        \label{fig:two_separate_engines}
    \end{figure}

    \subsection*{Interlude -- Tree-states}
    \label{sec:interlude}

    In this section, we present a general family of protocols that can be used to obtain highly non-equilibrium states in a controlled manner using a set of available elementary two-level operations, which are realisable experimentally via (intensity-dependent) Jaynes-Cummings interactions and collisional models~\cite{Lostaglio2018elementarythermal}. We refer to the resulting states as \textit{tree-states}, due to their relation to the underlying graphs, which will be elucidated below.
    Let us highlight that the following protocols may not provide extremal states for engines using ETO.
    Nevertheless, they provide interior bounds on the set of achievable states in a similar spirit as the states proposed in~\cite{KAMIL}. Furthermore, in below-considered examples we found the former to be a more precise than the latter.
    
    Let us begin with the simplest example of a 3-level system with an associated Gibbs state~$\gamma$. 
    First, we consider a subroutine 
    in which an engine repeatedly performs thermal swaps between two levels $0$ and $1$ corresponding to two initial populations $p^{(0)}_{0} = \gamma_0$, $p^{(0)}_{1} = \gamma_1$ and $p^{(0)}_{2} = \gamma_2$. 
    After a large number of rounds, the distribution on those two levels is driven to, up to exponential precision, $(p^{(0)}_{0},p^{(0)}_{1})\to (p^{(0)}_{0} + p^{(0)}_{1}) \vb{\tilde{\gamma}} :=(p^{(1)}_{0}, p^{(1)}_{1})$, 
    with $\tilde{\vb{\gamma}}$ given by eq.~\eqref{eq:extreme_qubit_pops} applied in a slight notation abuse in restriction to the $(0,1)$ levels of the 3-level system under consideration.
    Next, one can perform a similar subroutine between levels $0$ and $2$ with populations $p^{(1)}_{0}$ and $p^{(1)}_{2} = p^{(0)}_{2}$ driving them to $(p^{(1)}_{0}, p^{(1)}_{2}) \to (p^{(1)}_{0}+ p^{(1)}_{2}) \vb{\tilde{\Gamma}}:= (p^{(2)}_{0}, p^{(2)}_{2})$, where, once again, we introduce slight abuse of notation by considering $\vb{\tilde{\Gamma}}$ in restriction to $(0,2)$ levels of the system. 
    By applying the above subroutines in a loop one obtains a final state in which relations between final populations are given, up to an exponentially decaying error, by
    \begin{equation}
    \begin{aligned}
    \lim_{n\rightarrow\infty} \frac{p_0^{(n)}}{p_1^{(n)}} = \frac{\tilde{\gamma}}{1-\tilde{\gamma}} & = \frac{\left(1-e^{\beta_c \Delta E_{10}}\right) e^{\beta_h\Delta E_{10}}}{1-e^{\beta_h\Delta E_{10}}}\,, \\
    \lim_{n\rightarrow\infty} \frac{p_0^{(n)}}{p_2^{(n)}} = \frac{\tilde{\Gamma}}{1 - \tilde{\Gamma}} & = \frac{\left(1-e^{\beta_h \Delta E_{20}}\right) e^{\beta_c\Delta E_{20}}}{1-e^{\beta_c\Delta E_{20}}}\,,
    \end{aligned}
    \end{equation}
    where $\Delta E_{ij} = E_i - E_j$. Thus, we create a state with a large proportion of the population ``driven'' to a certain level.
    In this exemplary scheme, we say that there is \textit{a cooling coupling} between levels  $0$ and $1$ and \textit{a heating coupling} between levels $0$ and $2$.

    Using the restricted three-level example presented above, we may proceed to a fully 
    general engine with a total dimension equal $d$. 
    Let us consider a $d$-vertex graph, with one-to-one correspondence between vertices and energy levels of the system under consideration. The set of directed edges will correspond to a select subset of allowed two-level operations (couplings), which we use to arrive at a stationary state, and the directions describe the type of coupling. Below we present several desired properties for such a graph.
    
    \begin{enumerate}
        \item \textbf{No same-level couplings}: As the aim of the procedure is to drive an initial state as far from thermal equilibrium as possible, it is natural to exclude couplings between populations with the same energy, as such couplings would lead only to periodic oscillation.
        \item 
        \textbf{Acyclicity:} 
        If there were some cycles in the graph, then for any two coupled levels in the cycle, there would be two ways to calculate the ratio of populations between them: direct one and the one around the cycle. If those two ways coincide, then direct coupling is unnecessary and can be removed. If this is not the case, then the state described by such a graph would not be stationary but constantly driven around inside a set of achievable states.
        \item 
        \textbf{Connectedness:} 
        Consider, on the contrary, a state $\vb{p}$ corresponding to a graph consisting of at least two disconnected components, corresponding to a subset $I$ of indices and its complement $\overline{I}$. It is straightforward to see that
        \begin{align}
            \sum_{i\in I} p_i = \sum_{i\in I} \gamma_i,&&
            \sum_{i\in \overline{I}} p_i = \sum_{i\in \overline{I}} \gamma_i.
        \end{align}
        In other words, the populations corresponding to $I$ are in ``collective equilibrium'' with the remaining energy levels. Thus, there exists a pair of levels $i\in I,\, j \in \overline{I}$ for which introduction of coupling drives the state further away from equilibrium.
    \end{enumerate}

    Inferring from the above properties, one can easily notice that the desired graphs are trees.
    To be precise, let $G$ be a graph with vertices representing all energy levels and edges representing all couplings between levels with different energies allowed by engine theory. Then, all desired graphs are tree subgraphs of $G$, containing all vertices -- spanning trees. 

    \begin{figure}[h]
    \centering
    \begin{tikzpicture}[main_node/.style={circle,fill=black!10,draw,minimum size=1em,inner sep=3pt]}]
    
    \node[main_node] (1) at (0,0) {$p_{2}$};
    \node[main_node] (2) at (1, -1.5) {$p_{1}$};
    \node[main_node] (3) at (0, -3) {$p_{0}$};
    
    \draw (1) -- (2) -- (3) -- (1);
    \end{tikzpicture}
    ~\hspace{1.5cm}
    \begin{tikzpicture}[main_node/.style={circle,fill=black!10,draw,minimum size=1em,inner sep=3pt]},
    ->, >={Triangle[length=6pt,width=5pt]}]
    \node[main_node] (1) at (0,0) {$p_{2}$};
    \node[main_node] (2) at (1, -1.5) {$p_{1}$};
    \node[main_node] (3) at (0, -3) {$p_{0}$};
        
    \draw[->] (1)--(2);
    \draw[->] (3)--(1);
    \end{tikzpicture}    
    \caption{Graph representing possible couplings between populations of qutrit thermal engine (left) and graphical representation of exemplar qutrit tree state (right). For each used coupling we indicate by arrow the direction in which the population was driven.}
    \label{qtrut_trees}
    \end{figure}
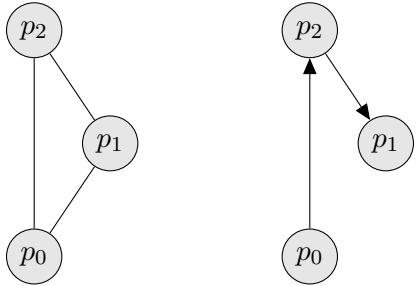

    With that observation, we may present the procedure to construct such highly athermal states, which, due to the underlying graphs, we call \textit{tree-states}.
    First, we calculate all spanning trees of the graph of $G$, whose edges represent all non-trivial couplings. Algorithms for the generation of spanning trees are widely known, eg.~\cite{SpanningTreesGen}.
    Next, for each spanning tree, we consider all $2^{k}$ combinations of ratios of populations between coupled levels (each coupling is driven towards $\tilde{\vb{\gamma}}$ or $\tilde{\vb{\Gamma}}$ like state as given in~\eqref{eq:extreme_qubit_pops} restricted to a selected pair of levels), where $k = d - 1$ is the number of edges (couplings) in the tree and $d$ the dimension of the system. For each combination, we express each population as $p_{0}$ times the product of all ratios along the path from $p_{0}$ to the given population.  Finally, for each state, we determine the value of $p_{0}$ by the normalization condition $\sum_{i} p_{i} = 1$.

    In the following sections, we will use tree states to obtain lower bounds for different thermal engines.
    Finally in Appendix~\ref{App:h_tree_&_Kamil} we extend the construction of tree states above elementary operations using the notion of hyper trees.

    \subsection{Joint bath temperature control}

    In the following three models, we assume that the constituent qubits are transferred between the baths synchronously -- meaning that we assume no means of transferring them between baths independently. It may be a result of both systems prepared in a common, extended environment at an equilibrium, which allows for control of ambient temperature significantly faster than the timescale required for our working medium -- the qubits -- to thermalise.

    We begin by considering a comparison between LTOCC and ETO, two frameworks for athermality manipulation for which, in discussed examples, free operations consist exclusively of two-level operations. Here, we demonstrate the application of tree-states to ETO and 2-round LTOCC, while showcasing the limitation for 1-round LTOCC. Then we move on to full thermal operations, where we stage analytically derived critical temperatures against the numerical results.
    
        \subsubsection{LTOCC vs ETO}

        In this section we will compare two restricted versions of thermal operations -- local thermal operations and classical operations (LTOCC) and elementary thermal operations (ETO). Note that LTOCC may be an applicable restriction if the experimental setup allows for separate local manipulation of constituent systems, but has no way to address transitions which involve more than one party. However, given that measurement of a subsystem and the result used to effect conditional operations on another subsystem can be performed significantly faster than thermalisation timescale, one can implement LTOCC. On the other hand, ETO have been shown to be realisable by (intensity-dependent) Jaynes-Cummings model, thus rendering them as a physically feasible lab model of thermal operations~\cite{jaynes-cummings_original, Shore1993}. 
        
        We begin by studying a 1-round realisation of LTOCC protocol, in which a single measure-and-condition operation is allowed per round. Here, we assume that Alice and Bob take turns in measuring, ie. during a cold stroke Alice measures, sends the information to Bob, and then he performs the operations, while during a hot stroke, it is Bob who measures, and sends the information over to Alice. This setup prevents us from applying the tree-state formalism for assessing, how well an engine based on such protocol would perform, since neither party can freely couple two energy levels in two different temperatures. As such, we resort to a numerical approach by approximating the limit of a number of strokes by considering a finite number of strokes with convergence conditions; details of the approach have been described in the Appendix~\ref{app:numerical_proc}. The results are presented in Fig.~\ref{fig:LTOCC_v1_1r_vals}. Compared to two separate qubit engines ground-state population $\mathcal{P}_G$ is qualitatively indistinguishable. The excited state population $\mathcal{P}_E$ exhibits more advantage for LTOCC protocols, with a visible change of convexity of level sets. However, the most drastic qualitative change can be seen for the negativity $\mathcal{N}_{max}$. Here, there are several features to be noted: First, we stress reduction in the no-entanglement set delineated by the dashed red line, especially for very low temperatures, $e^{-\beta} \approx 0$ which we discuss later for the LTOCC engine with a larger number of rounds. Next, we note the appearance of almost maximal $\mathcal{N}_{max}$ achievable whenever one of the baths is at very low temperature, $e^{-\beta}\approx 0$, which was not present for separate engines. Finally, we note a surprising qualitative symmetry under $e^{-\beta_c} \leftrightarrow 1 - e^{-\beta_h}$ exchange. This, however, is only approximate, which is easily checked by appropriate reflection and overlay of the numeric data.

        Extended to two rounds of memoryless LTOCC per stroke, the situation becomes more convoluted; naive counting leads to the number of potential operations per engine cycle as large as $32^2 = 2^{10}$, out of which not all are extremal operations for the joint system. However, by leveraging the fact that we consider engines with an arbitrary number of strokes one may partially simplify the situation.         
        
        \begin{figure}[h]
            \centering
            \includegraphics[width=1\linewidth]{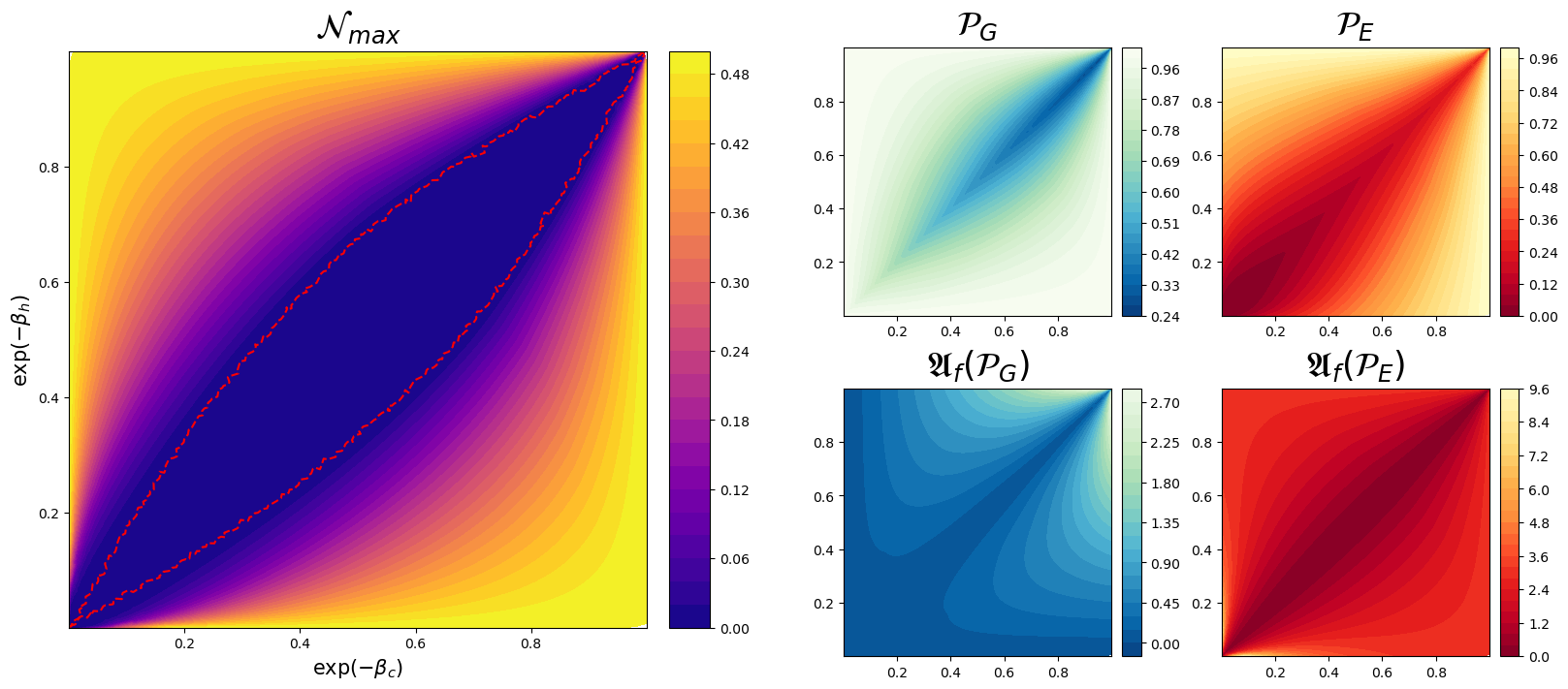}
            \caption{\textbf{One-round LTOCC engine at identical temperatures:} Introduction of a single-round LTOCC protocol with operations at each stroke taken at a constant temperature allows significant advantage with respect to the negativity generation -- we observe possibility of generating maximal $\mathcal{N}_{max}$ also in the limit $\beta_c\rightarrow \infty$. Excited state population $\mathcal{P}_E$ seems to be affected by change in convexity of level set boundaries, and ground state population $\mathcal{P}_G$ does not experience qualitative changes with respect to separate qubit engines. Ground state advantage $\mathfrak{A}_f(\mathcal{P}_G)$ exhibits qualitative similarity to the one calculated for separate engines, while $\mathfrak{A}_f(\mathcal{P}_E)$ grows significantly for the low temperature regime.}
            \label{fig:LTOCC_v1_1r_vals}
        \end{figure}

        Let us discuss one engine stroke step by step. First Alice performs a measurement and sends its result to Bob, who applies conditional thermal operation. Next Alice can perform thermal post-processing on her subsystem, which is followed by thermal operation conditioned by Bob's measurement. However, those two consecutive thermal operations on Alice's subsystem can be merged into one making her post-processing redundant. Finally, the stroke ends with Bob's thermal post-processing, after which the temperature is changed. However, if we skip the new temperature (set all operations to identity) and come back to the original one, we may map Bob's postprocessing into (trivially) conditioned thermal operation on Alice's measurement in the next stroke and only then proceed to the operations in new temperatures. 
        In such a way one may disregard all thermal post-processing in the protocols at the cost of making them at most twice longer.

        By the same token, we may show a very useful statement

        \begin{lem}\label{Lem:2_vs_multi_LTOCC}
            Two-round LTOCC thermal engine is equivalent to any LTOCC thermal engine with a larger number of rounds, in the regime of arbitrarily long protocols.
        \end{lem}

        \begin{proof}
        The inclusion in one direction is trivial, since the two-round LTOCC is a subset of $n$-round LTOCC with $n \geq 2$. The inclusion in the opposite direction follows from a similar argument as above. Each $n$-round LTOCC can be implemented in at most $n/2$ strokes of a 2-round LTOCC engine by performing consecutive rounds if the temperature is correct and skipping the steps with inappropriate temperature.
        \end{proof}

        Leveraging above discussion, we can apply the tree-state formalism to provide accessible interior bounds. In this scenario, the accessible couplings follow from the conditioned thermal operation and are presented in Figure~\ref{Starting_graph} (left). The extremal tree states are presented graphically in Figure~\ref{Extremal_tree_states} (left side) which leads us to the following lower bounds for $p_{11}$ and $p_{00}$ populations for LTOCC engines: 
        \begin{subequations}
        \label{tree_ltocc_heating_cooling}
            \begin{align}
            &\max_{\text{LTOCC tree states}} p_{11} =\frac{\left(e^{\beta_c}-1\right){}^3 e^{\beta_h}}{\left(e^{\beta_c+\beta_h}-1\right) \left(e^{2 \beta_c}+\left(1-e^{\beta_c} \left(e^{\beta_c}+2\right)\right) e^{\beta_h}+e^{2 \left(\beta_c+\beta_h\right)}\right)} ,\\
            &\max_{\text{LTOCC tree states}} p_{00} = \frac{\left(e^{\beta_c}-1\right){}^3 e^{2 \beta_h}}{\left(e^{\beta_c+\beta_h}-1\right) \left(e^{\beta_h} \left(e^{\beta_c} \left(e^{\beta_c}-2\right)+e^{\beta_h}-1\right)+1\right)}.
            \end{align}
        \end{subequations}
        We let ourselves omit the expression for the lower bound for maximum negativity (obtained with extra operation in energy degenerate subspace) and for advantages due to a highly convoluted formula, referring the reader to Fig.~\ref{fig:treestates_vals} for its visual representation instead. 

        \begin{figure}[h]
        \centering
        \begin{tikzpicture}[main_node/.style={circle,fill=black!10,draw,minimum size=1em,inner sep=3pt]}]
    
        \node[main_node] (1) at (0,0) {$p_{11}$};
        \node[main_node] (2) at (-1, -1.5)  {$p_{01}$};
        \node[main_node] (3) at (1, -1.5) {$p_{10}$};
        \node[main_node] (4) at (0, -3) {$p_{00}$};
    
        \draw (1) -- (2) -- (4) -- (3) -- (1);
        \end{tikzpicture}
        ~\hspace{1.5cm}
        \begin{tikzpicture}[main_node/.style={circle,fill=black!10,draw,minimum size=1em,inner sep=3pt]}]
    
        \node[main_node] (1) at (0,0) {$p_{11}$};
        \node[main_node] (2) at (-1, -1.5)  {$p_{01}$};
        \node[main_node] (3) at (1, -1.5) {$p_{10}$};
        \node[main_node] (4) at (0, -3) {$p_{00}$};
    
        \draw (1) -- (2) -- (4) -- (3) -- (1);
        \draw (1) -- (4) ;
        \end{tikzpicture}
        \caption{\textbf{Tree-state skeleton graphs:} Graph representing possible couplings between populations for 2-qubit LTOCC engine (left) and ETO engine (right). The starting point to construct extremal tree states.
        }
        \label{Starting_graph}
        \end{figure}
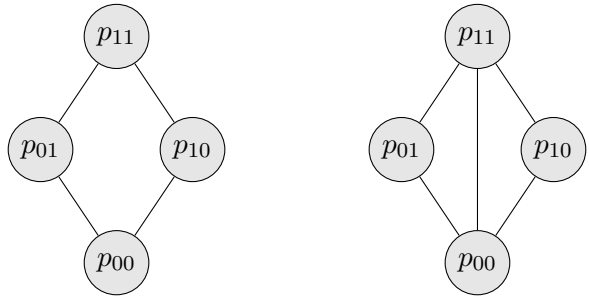

        \begin{figure}[h]
\centering
\begin{tikzpicture}[
    scale=.65, 
    main_node/.style={circle,fill=black!10,draw,minimum size=1em,inner sep=3pt},
    ->, >={Triangle[length=6pt,width=5pt]}
  ]

  \newcommand{\drawnodes}{
    \node[main_node] (1) at (0,0) {};
    \node[main_node] (2) at (-1, -1.5) {};
    \node[main_node] (3) at (1, -1.5)  {};
    \node[main_node] (4) at (0, -3) {};
  }

  \begin{scope}[xshift=0cm, yshift=0cm]
    \drawnodes
    \draw[->] (2)--(1);
    \draw[->] (4)--(2);
    \draw[->] (3)--(4);
  \end{scope}

  \begin{scope}[xshift=-3cm, yshift=-4.5cm]
    \drawnodes
    \draw[->] (3)--(1);
    \draw[->] (1)--(2);
    \draw[->] (2)--(4);
  \end{scope}

  \begin{scope}[xshift=3cm, yshift=-4.5cm]
    \drawnodes
    \draw[->] (3)--(1);
    \draw[->] (2)--(1);
    \draw[->, red, dashed] (1)--(4);
  \end{scope}

  \begin{scope}[xshift=-6cm, yshift=-9cm]
    \drawnodes
    \draw[->] (4)--(3);
    \draw[->] (3)--(1);
    \draw[->] (1)--(2);
  \end{scope}

  \begin{scope}[xshift=-2cm, yshift=-9cm]
    \drawnodes
    \draw[->] (1)--(3);
    \draw[->] (3)--(4);
    \draw[->] (4)--(2);
  \end{scope}

  \begin{scope}[xshift=2cm, yshift=-9cm]
    \drawnodes
    \draw[->] (3)--(1);
    \draw[->, red, dashed] (1)--(4);
    \draw[->] (4)--(2);
  \end{scope}

  \begin{scope}[xshift=6cm, yshift=-9cm]
    \drawnodes
    \draw[->] (3)--(4);
    \draw[->, red, dashed] (4)--(1);
    \draw[->] (1)--(2);
  \end{scope}

\end{tikzpicture}

\caption{\textbf{Optimal tree-state graphs:} Graphical representations of tree-states of two-qubit LTOCC engine and ETO engine, which obtains maximal $p_{11}$ population (top row), maximal $p_{00}$ population (middle row), and maximal entanglement (bottom row). For each coupling, we indicate by the arrow in which direction the populations were driven. Note that the graphs containing red dashed arrows cannot be realised by LTOCC, and thus are restricted to ETO protocols. Additionally, due to their essential equivalence, the graphs with exchanged populations $p_{10}\leftrightarrow p_{01}$ are omitted. 
}
\label{Extremal_tree_states}
\end{figure}
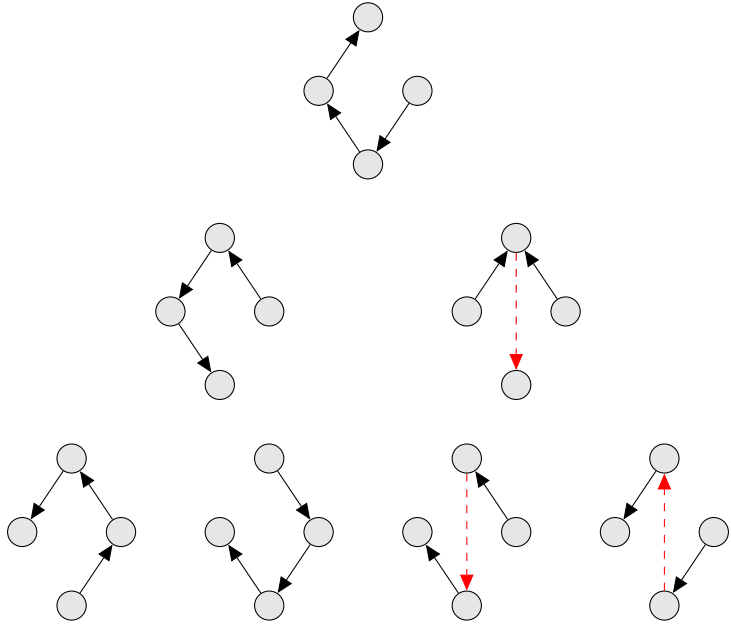

        Finally, assuming complete control over two-level transitions, we may proceed to analyse an engine operating using ETO at each stroke. Note that since conditioning may be implemented as a thermal operation~\cite{bistron2024local}, hence as long as LTOCC utilizes only elementary operations engine protocols based on them form a subset of ETO engine protocols. 
        Thus, we have the following inclusion.

        \begin{lem}\label{Lem:El_inclusion}
            The set of protocols possible to implement in the LTOCC engine using elementary conditioned operations forms a subset of ETO engine protocols.
        \end{lem}

        The alternative argument for the above result is the observation that the graph of couplings for the ETO engine, presented in Figure~\ref{Starting_graph} (right) for the 2-qubit engine, is a super-graph of couplings for LTOCC. The same holds also for the set of extremal tree states for elementary LTOCC and ETO, presented in Figure~\ref{Extremal_tree_states}: the former is a superset of extremal tree states for the latter. Since there are no new tree states that maximize $p_{11}$ population, the bound~\eqref{tree_ltocc_heating_cooling} still holds, however, a new state which may maximize $p_{00}$ population appears giving the result
        \begin{equation}
        \label{tree_states_ETO}
        \begin{aligned}
        \max_{\text{ETO tree states}} p_{00}  = \max \Bigg\{& \max_{\text{LTOCC tree states}} p_{00}  ~, \\
        &\frac{\left(e^{\beta_c}-1\right){}^2 \left(e^{\beta_c}+1\right) e^{2 \beta_h}}{e^{\beta_c} \left(e^{\beta_h} \left(e^{\beta_h} \left(e^{\beta_c} \left(e^{\beta_c}-1\right)+2 e^{\beta_h}-2\right)-2\right)+1\right)+1} \Bigg\} \,,
        \end{aligned}
        \end{equation}
        where for each term there exists a regime in which it dominates.
        Similarly, for ETO engine, we found $4$ different tree states giving maximal entanglement after combining them with extra operation in energy-degenerate subspace, each of them dominating in a different regime. But we let ourselves omit explicitly formulate due to a long and convoluted structure.

        \begin{figure}[t]
            \centering
            \includegraphics[width=1\linewidth]{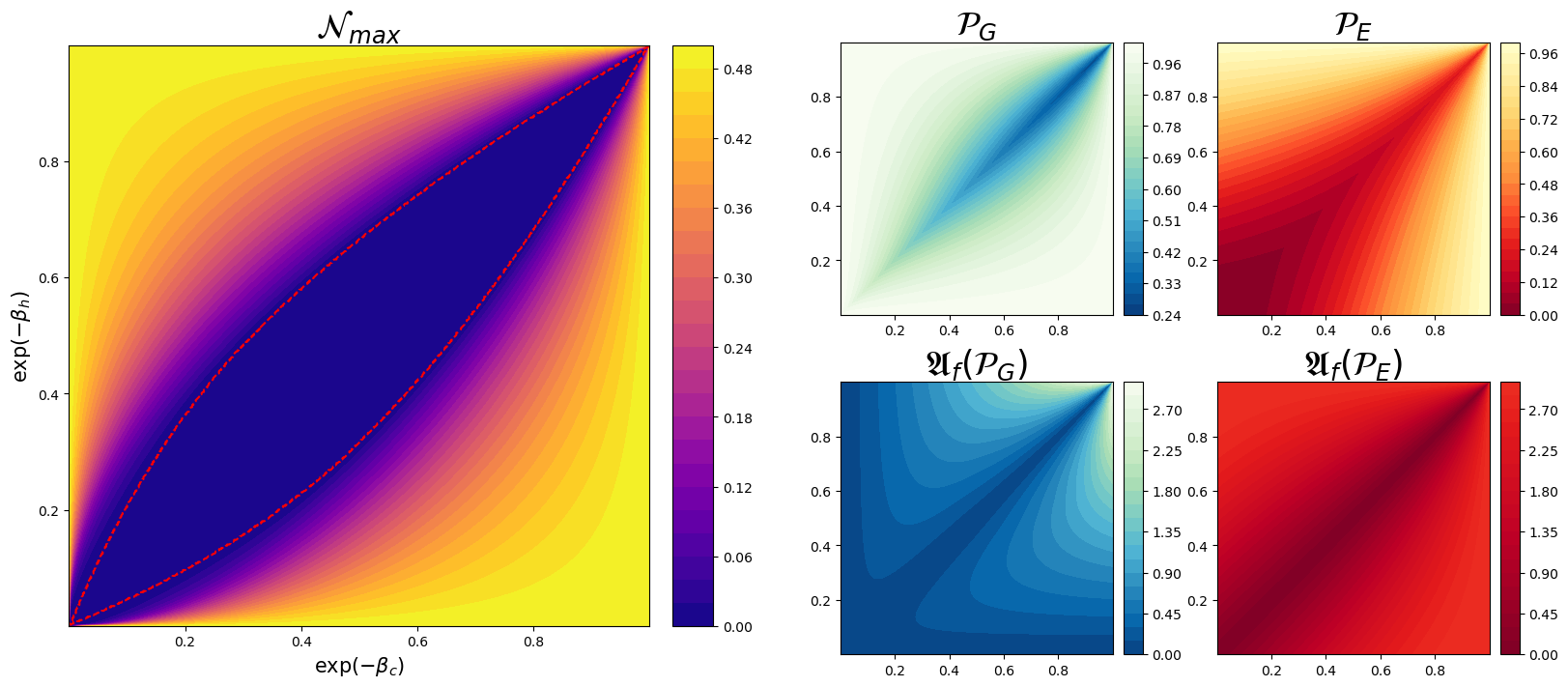}
            \caption{\textbf{Tree-state ETO engine:} By restricting thermal engines' free states to tree-states as described in the \hyperref[sec:interlude]{Interlude}, we obtain analytical interior bounds on achievable states and, as a consequence, lower bounds on $\mathcal{N}_{max},\,\mathcal{P}_G$ and $\mathcal{P}_E$ for full ETO and LTOCC engines. Note that ground state advantage $\mathfrak{A}_f(\mathcal{P}_G)$ is very close to one-round LTOCC engine (Fig.~\ref{fig:LTOCC_v1_1r_vals}), while $\mathfrak{A}_f(\mathcal{P}_E)$ fails to show similar growth in low temperatures, exhibiting similarity to separate qubits (Fig.~\ref{fig:sep_qubits}).}
            \label{fig:treestates_vals}
        \end{figure}

        \begin{figure}[H]
            \centering
            \includegraphics[width=0.76\linewidth]{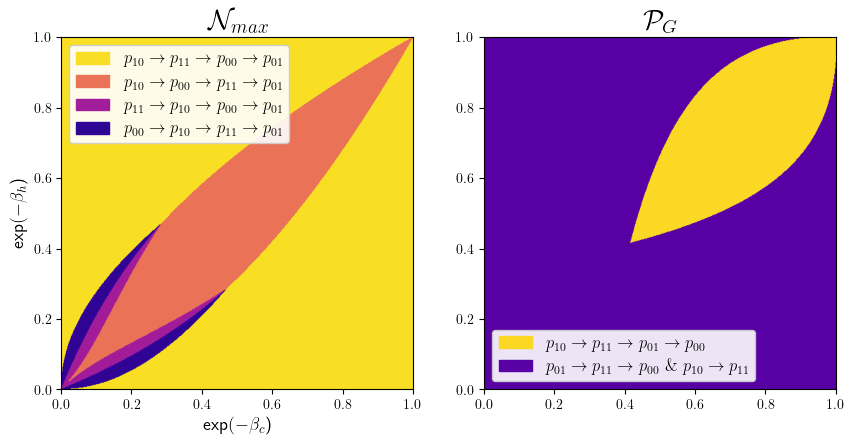}
            \caption{\textbf{Tree-state optimality regions:} The regions in which different tree state protocols for ETO engine are optimal, with the left corresponding to the entanglement generation and the right to cooling. 
            Different colours from navy-blue to yellow indicate consecutive tree states presented in Figure~\ref{Extremal_tree_states} from left to right. The directions of population flow in tree states are also denoted in the legends.}
            \label{fig:treestates_reg}
        \end{figure}

        Finally, one can understand the qualitative difference between the abilities of LTOCC (and ETO) engine compared to two separate qubit engines, by noticing that the optimal protocols of the latter, discussed in Section~\ref{two_qubit_simple}, correspond to simple tree states, see Fig.~\ref{trees_for_2_separate_q}.

        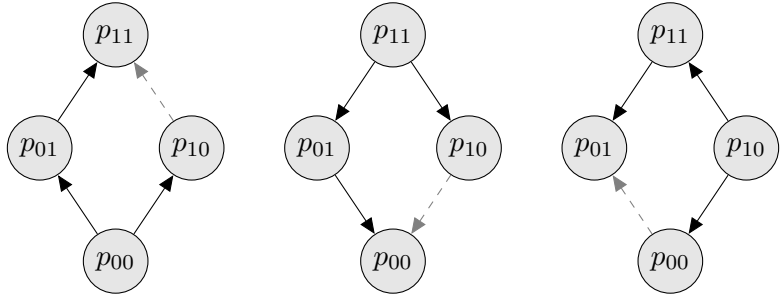
\begin{figure}[h]
            \centering
            
            \begin{tikzpicture}[main_node/.style={circle,fill=black!10,draw,minimum size=1em,inner sep=3pt]},
            ->, >={Triangle[length=6pt,width=5pt]}]
            \node[main_node] (1) at (0,0) {$p_{11}$};
            \node[main_node] (2) at (-1, -1.5)  {$p_{01}$};
            \node[main_node] (3) at (1, -1.5) {$p_{10}$};
            \node[main_node] (4) at (0, -3) {$p_{00}$};
        
            \draw[->] (4)--(2);
            \draw[->] (2)--(1);
            \draw[->] (4)--(3);
            \draw[->,gray,dashed] (3)--(1);
            \end{tikzpicture}
            ~\hspace{0.4cm}
            \begin{tikzpicture}[main_node/.style={circle,fill=black!10,draw,minimum size=1em,inner sep=3pt]},
            ->, >={Triangle[length=6pt,width=5pt]}]
            \node[main_node] (1) at (0,0) {$p_{11}$};
            \node[main_node] (2) at (-1, -1.5)  {$p_{01}$};
            \node[main_node] (3) at (1, -1.5) {$p_{10}$};
            \node[main_node] (4) at (0, -3) {$p_{00}$};
        
            \draw[->] (1)--(2);
            \draw[->] (2)--(4);
            \draw[->] (1)--(3);
            \draw[->,gray,dashed] (3)--(4);
            \end{tikzpicture}
            ~\hspace{0.4cm}
            \begin{tikzpicture}[main_node/.style={circle,fill=black!10,draw,minimum size=1em,inner sep=3pt]},
            ->, >={Triangle[length=6pt,width=5pt]}]
            \node[main_node] (1) at (0,0) {$p_{11}$};
            \node[main_node] (2) at (-1, -1.5)  {$p_{01}$};
            \node[main_node] (3) at (1, -1.5) {$p_{10}$};
            \node[main_node] (4) at (0, -3) {$p_{00}$};
        
            \draw[->] (3)--(1);
            \draw[->] (1)--(2);
            \draw[->] (3)--(4);
            \draw[->,gray,dashed] (4)--(2);
            \end{tikzpicture}
            
            \caption{\textbf{Flow redundancy:} Graphical representations of optimal protocols in 2 separate qubit engines to obtain maximal $p_{11}$ population, maximal $p_{00}$ population and maximal entanglement (from left to right) as tree-states. For each coupling we indicate by arrow in which direction the populations were driven. Grey arrows correspond to redundant couplings, which can be omitted resulting in tree states. 
            }
            \label{trees_for_2_separate_q}
        \end{figure}

        Notice that for each considered extremal state of 2 qubit engines, there exists a tree state of LTOCC engine with the same couplings, but with direction of population's driving switched on one coupling, due to greater freedom. This switch in each case allows one to reduce an undesired population, $p_{10}$ for warming and cooling and $p_{11}$ for entanglement, by a factor $\frac{1-\vb{\tilde{\gamma}}}{\vb{\tilde{\gamma}}}\frac{\vb{\tilde{\Gamma}}}{1 - \vb{\tilde{\Gamma}}}$ (compared to other populations), without affecting other rations between populations. In the limit $e^{-\beta_c}\approx 0$ the above factor has a pole of a form $e^{-\beta_h}/(e^{-\beta_c}(1 -e^{-\beta_h})^2)$ resulting in a strong suppression of undesired population. This effect is most clearly seen in the case of entanglement generation, since the minimization of $p_{11}$ population results in an almost sharp state in the subspace $(00,11)$, which makes the total state almost maximally entangled.

        We close consideration of LTOCC and ETO engines with equal bath temperatures by demonstrating numerical results.
        Due to their qualitative similarity to Figure~\ref{fig:treestates_vals} we defer the corresponding Figures~\ref{fig:LTOCC_v1_2r_vals} and~\ref{fig:ETO_vals}, presenting the capacity of LTOCC and ETO engines, to the Appendix~\ref{app:extra_figures}, whereas here we demonstrate differences between the full action of an engine approximated numerically with the analytical bounds on achievable values as obtained from the tree-states. The differences are presented qualitatively in Figure~\ref{fig:tree_LTOCC_ETO_comparison}. Note that comparison does not provide surprises -- first, we see that LTOCC tree states provide a valid interior bound for both full LTOCC and ETO engines, which can be seen in the left-side plots. Next, we see in the right-side plots that ETO tree states are valid interior bounds for full ETO engine (bottom right), while, somewhat predictably, they provide an advantage over full LTOCC engine. In addition, top-bottom comparison reveals differences between full ETO and LTOCC engines, which would be otherwise hard to spot from the figures presented in the Appendix~\ref{app:extra_figures}. 
        \begin{figure}[H]
            \centering
            \includegraphics[width = 0.48\linewidth]{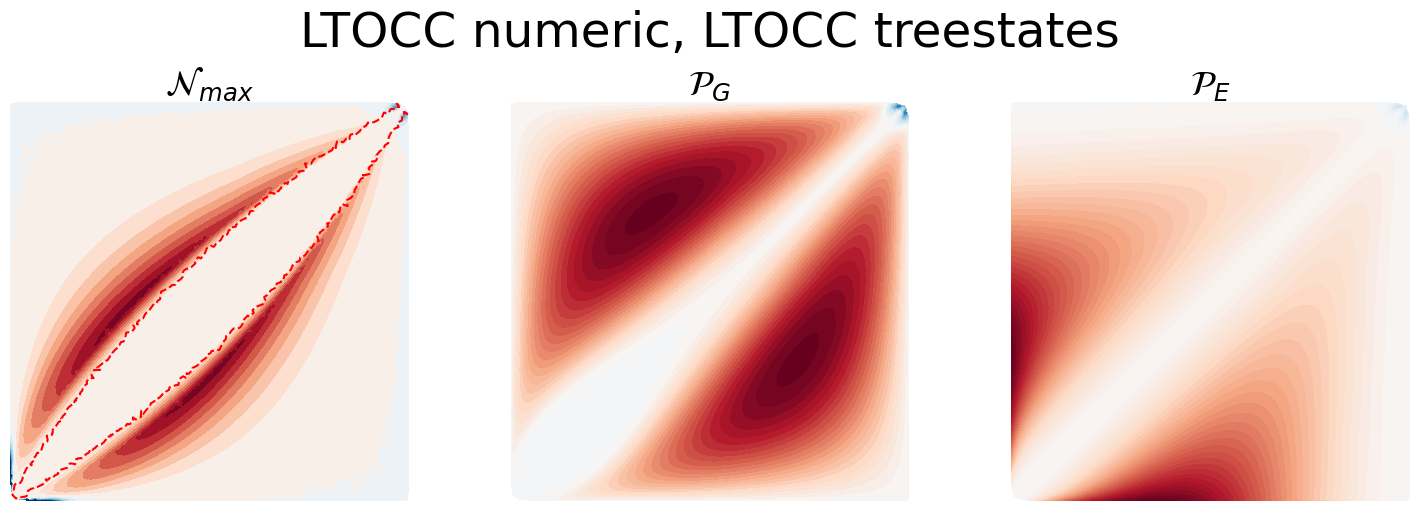}
            \hspace{.02\linewidth}
            \includegraphics[width = 0.48\linewidth]{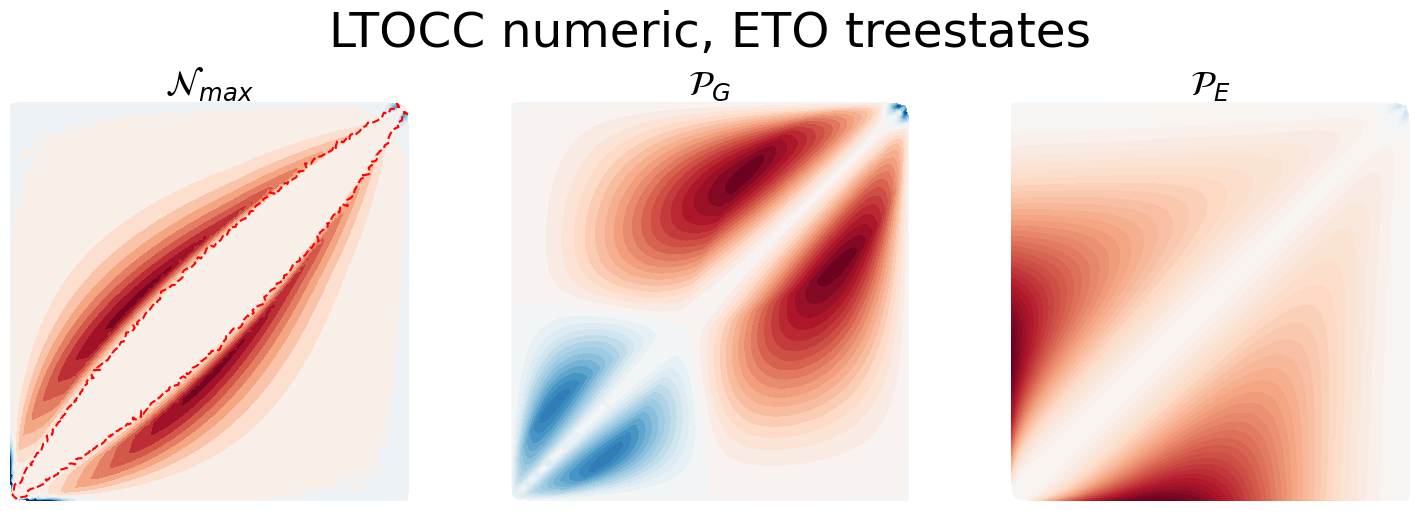}
            
            \vspace{1em}
            
            \includegraphics[width = 0.48\linewidth]{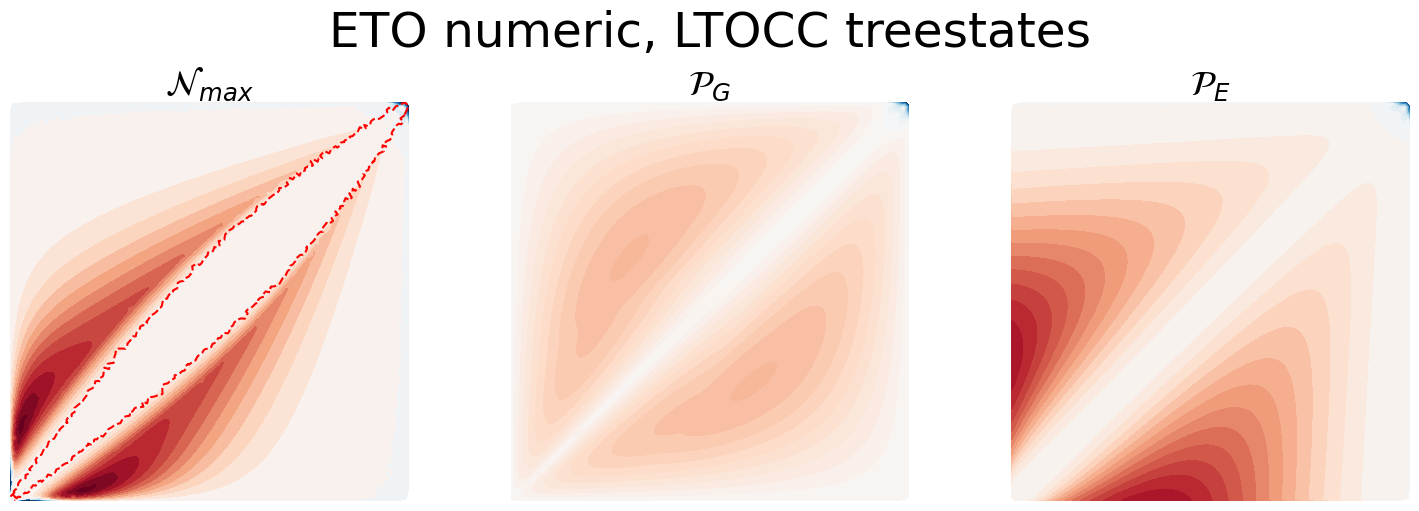}
            \hspace{.02\linewidth}
            \includegraphics[width = 0.48\linewidth]{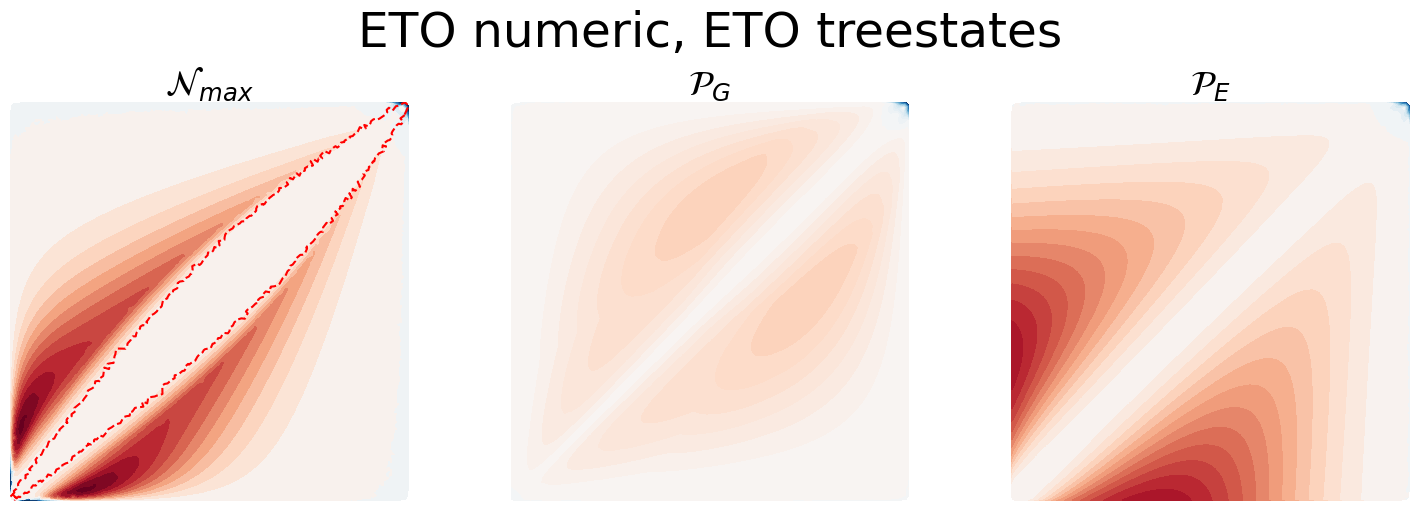}            \caption{
                \textbf{Tree-state vs. full engine comparison:} In the above four sets of qualitative plots we compare tree-state performance with complete engine action for both ETO and LTOCC and, for completeness, we provide a cross-comparison. Red and blue regions correspond to full engine and tree-state advantage, respectively. Note that blue region appears only for the ground state when comparing full LTOCC engine with ETO tree-states; this is to be expected, as ETO$\supset$LTOCC.
            }
            \label{fig:tree_LTOCC_ETO_comparison}
        \end{figure}

        \subsubsection{Two-qubit thermal operations}

        The last scenario in which both qubits interact with the thermal bath in the same temperature at any given time is modelled by full thermal operations at each stroke of the engine. 
        We stress that in this model, and any other one discussed above, there exists an upper bound for the maximally excited population given by $p_{d,d} \leq e^{- \beta_{c}(E_d - E_{d-1}) }$  where $d$ is a dimension of each local subsystem~\cite{KAMIL}. In our particular case of a two-qubit engine, it simplifies to $p_{11}\leq e^{-\beta_c}$. 
        
        Although the extension of tree states to non-elementary operation is possible as described in Appendix~\ref{App:h_tree_&_Kamil}, in the present case they do not tighten the constraint obtained from tree states using elementary thermal operations for a 2-qubit engine. Moreover, the far from equilibrium states constructed in Proposition 5 from ~\cite{KAMIL} are in principle applicable in this case, but they do not lead to any significant improvement in approximation as well, see Appendix~\ref{App:h_tree_&_Kamil}. An exhaustive analytical approach is also not applicable, as the extreme operations available are highly dependent on the temperature. In particular, one can identify three critical temperatures, corresponding solutions of equations
        \begin{subequations}\label{eq:crit_temps}
            \begin{align}
                2e^{-\beta_1} + e^{-2\beta_1} = 1 &\quad\Longrightarrow\quad e^{-\beta_1} = \sqrt{2}-1 = \sigma^{-1},\\
                2e^{-\beta_2} = 1 &\quad\Longrightarrow\quad e^{-\beta_2} = \frac{1}{2},\\
                e^{-\beta_3} + e^{-2\beta_3} = 1 &\quad\Longrightarrow\quad e^{-\beta_3} = \frac{\sqrt{5}-1}{2} = \varphi^{-1}.
            \end{align}
        \end{subequations}
        The appearance of the golden ratio $\varphi$ and the silver ratio $\sigma$, although fully expected, is a curiosity that deserves a highlight. Additionally, for convenience, we define $\beta_0 = 0$ and $\beta_4 = \infty$. 
        The existence of three critical temperatures entails $4^2$ different temperature regimes (together with boundary regimes, in which at least one of the temperatures is critical), in which different sets of extreme operations are available at each engine cycle, with their number on the order of $(4!)^2 = 24^2$, thus rendering the complete analytical analysis infeasible. 
        As such, we have mostly resorted to extended numerical methods to investigate in detail the power provided by full thermal operations. First, it is instructive to notice that the critical temperatures as outlined in~\eqref{eq:crit_temps} indeed influence the performance of the engines with respect to all considered metrics, dividing the temperature-parameter space into 16 full-dimensional regions, 24 critical line segments on which one of the temperatures is exactly critical, and 9 critical points, for which both baths are set at their respective critical temperatures. 
        The capacities of the thermal engine in all these regions are presented and described in Figure~\ref{fig:full_to_engine}. Further numerical study of TO engine, with an emphasis on critical behaviour, is presented in the Appendix~\ref{app:edge_detect}.
        For the sake of completeness, we present one more important relation between different types of thermal engines.

        \begin{lem}
            The set of protocols possible to implement in the LTOCC engine forms a subset of TO engine protocols.
        \end{lem}

        \begin{proof}
            The claim follows from exactly the same arguments as presented in Lemma~\ref{Lem:El_inclusion}.
        \end{proof}

        \begin{figure}[H]
            \centering
            \includegraphics[width=1\linewidth]{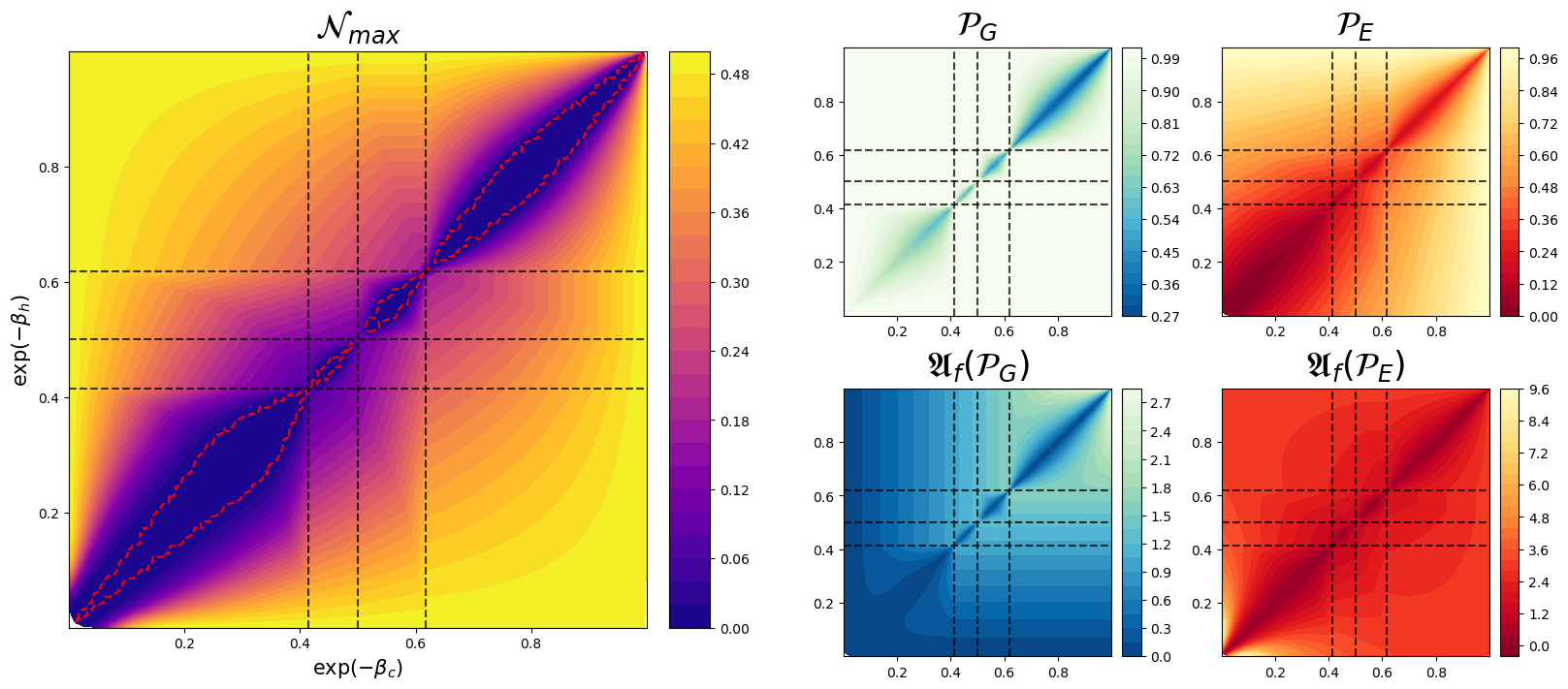}
            \caption{\textbf{Full TO engine:} By allowing arbitrary thermal operations on two qubits at each stroke, we find a significant achievable advantage. First, we note that the no-negativity regions, $\mathcal{N}_{max} = 0$, are restricted to $\beta_{i-1}\geq \beta_c,\beta_h \geq\beta_i$ for $i = 1,2,3,4$. The advantage is even more pronounced for ground-state population $\mathcal{P}_G$, which is not saturated only in the aforementioned regions. Finally, for the excited state population the appearance of critical temperatures is also visible, but no obviously describable properties can be elucidated. The corresponding advantage  $\mathfrak{A}_f(\mathcal{P}_E)$ shows growth over LTOCC model, with small values concentrated only around $\beta_c \approx \beta_h$.}
            \label{fig:full_to_engine}
        \end{figure}

    \subsection{Independent bath temperature control}

    In the following sections, we will focus on scenarios where we allow the qubits to be in contact with baths of different temperatures, while still interacting. Note, that the initial example -- two qubits without interaction -- did not differentiate between this situation and a case of an identical bath for two qubits, as in either scenario qubits could be driven to the same end states.

    We begin by considering LTOCC with unequal temperatures, showcasing that asymmetry arising in one-round-per-stroke protocols vanishes for two-rounds-per-stroke variant, which is shown to be equivalent to the equal temperature case. We next move on SLTO with their elementary variant limited to 2-level operations. There, we present that even an arbitrary small temperature difference leads to entanglement and ground state saturation, while achievable maximally excited state population is shown to be a faithful monotone under SLTO engine operations. 

        \subsubsection{LTOCC}

        For the distant laboratory paradigm one can easily envision a scenario where two laboratories are synchronising so that they exchange the roles of measured system and conditionally evolved system, while at the same time exchanging temperature. Thus, the measured constituent is always at the same temperature, no matter if it is on Alice's or Bob's side. This introduces a pronounced asymmetry between both sides, which can be easily seen in the numerics presented in Fig.~\ref{fig:LTOCC_1r_diff_temps}. Unfortunately, due to extensive complications related to the number of potential extreme operations and comparison between the respective states, the analytical approach becomes prohibitively difficult.  

        \begin{figure}[h]
            \centering
            \includegraphics[width=1\linewidth]{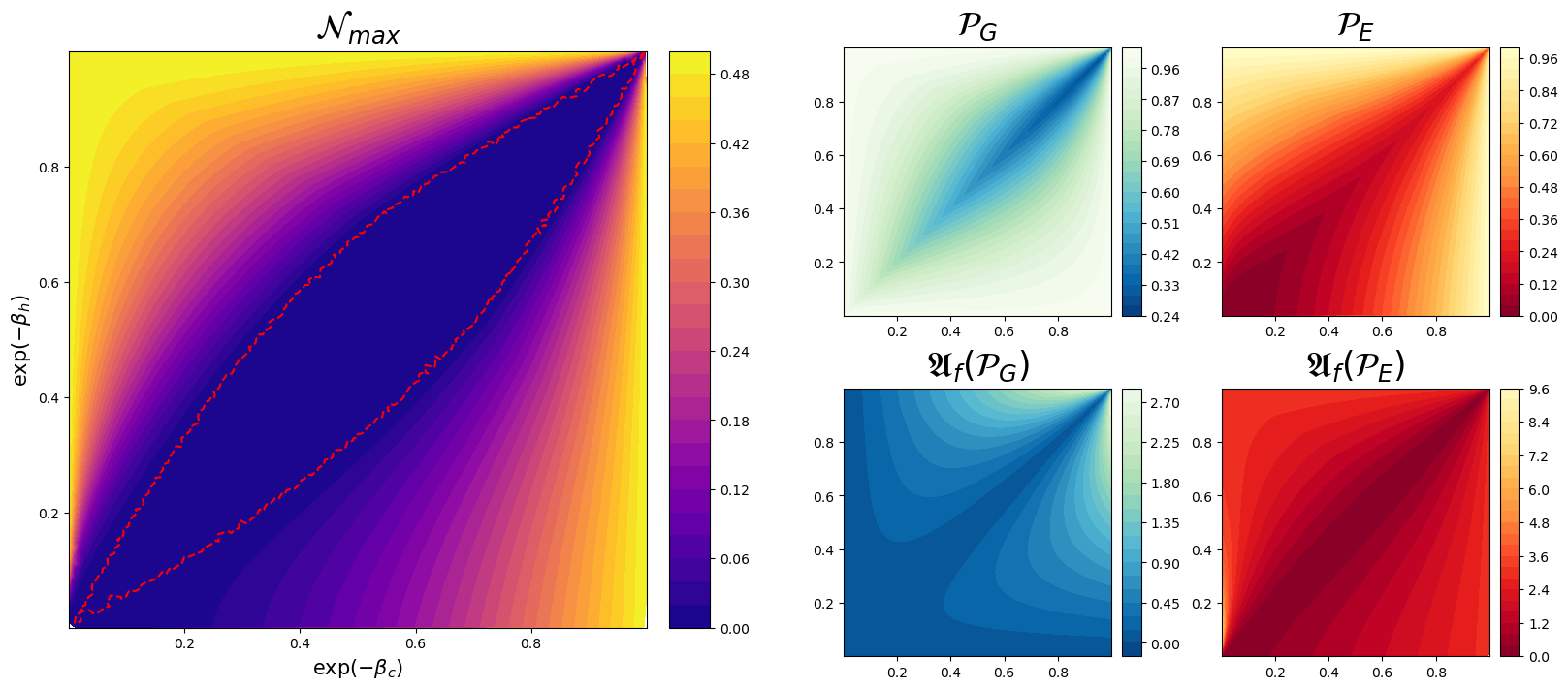}
            \caption{\textbf{One-round LTOCC engine with asymmetric temperatures:} in the case of LTOCC with temperature asymmetry at each stroke one party measures and post-processes at temperature $\beta_c$, while the other processes their system conditionally at $\beta_h$, and the roles are strictly tied to the temperature. This asymmetry is most pronounced for negativity, where for $\beta_h \geq \beta_c$ the level sets resemble the ones for separate qubit engines (Fig.~\ref{fig:sep_qubits}), while $\beta_h \leq \beta_c$ resembles the single-round LTOCC with equal temperatures (Fig.~\ref{fig:LTOCC_v1_1r_vals}). Both advantages $\mathfrak{A}_f(\mathcal{P}_G)$ and  $\mathfrak{A}_f(\mathcal{P}_E)$ exhibit similar asymmetry.
            }
            \label{fig:LTOCC_1r_diff_temps}
        \end{figure}

        Extending the above scenario to 2-round LTOCC with different temperatures, however, turns out to yield the same results as the temperature-symmetric case, in line with Lemma~\ref{Lem:2_vs_multi_LTOCC} and the accompanying discussion.
        By the same arguments as therein, thermal post-processing can be absorbed by conditioned thermal operations. Furthermore, in conditioned operations, the temperature of only one subsystem plays a role, so it doesn't matter if the conditioning system is in the same or different temperature. Since one performs conditioned thermal operation on both subsystems in both temperatures, we arrive at the following conclusion

        \begin{thm}
            Two-round LTOCC engine is equivalent to any other LTOCC engine with the same or larger number of rounds in the regime on arbitrary long protocols, irrespective of a chosen temperature control.
        \end{thm}
    
        \begin{proof}
        The equivalence between LTOCC engines in different temperature controls was discussed in the above paragraph. Furthermore, the equivalence between LTOCC engines with different numbers of rounds follows from exactly the same arguments as in Lemma~\ref{Lem:2_vs_multi_LTOCC}.
        \end{proof}

        \subsubsection{ESLTO and SLTO}\label{subsubsec:STLO}

        Finally, we move to the last pair of thermal engines directly connected to completely thermal processes, thermal engines based on Semilocal thermal operations (SLTO) and their elementary variant (ESLTO).  These engines turn out to be the most powerful of all engines considered so far. We start by presenting an upper bound for the engine performance and then present an explicit construction saturating this bound. Remarkably, to saturate the bound for SLTO engine, it is sufficient to use elementary SLTO operations. We also construct a faithful monotone for SLTO engine, which turns out to be proportional to engine efficiency as defined in~\eqref{eff_def}.
        In the following, we present a general result for a $d^2$-level system consisting of two identical subsystems with the same Hamiltonians and mark population on the level with local energies $E_n$ and $E_m$ in consecutive subsystems as $p_{n,m}$. 

        One stroke of SLTO engine corresponds to connecting the first subsystem with the first bath and the second with the second bath. Between rounds, the baths are exchanged with each other, effectively alternating the changing temperatures of the subsystems.

        We start by presenting a modification of Theorem 1 from~\cite{KAMIL} adjusted to SLTO engine.

        \begin{lem}\label{SLTO_upper_bound}
        Consider an SLTO engine operating on a pair of $d$-level subsystems with two different temperatures $\beta_h < \beta_c$.
        Then for each state $\vb{q}$ achievable from $\vb{p}$ in a possibly infinite amount of rounds, the following inequality holds:
        \begin{equation}
            q_{d,d} \leq \mathcal{M}_0(\vb{p}) := \max\{p_{d,d}, e^{-\beta_h(E_d - E_{d-1})} \}.
        \end{equation}
        Thus, $ \mathcal{M}_0(\vb{p}) := \max\{p_{d,d}, e^{-\beta_h(E_d - E_{d-1})} \}$ monotonically decrease under action of SLTO engine. \end{lem}

        \begin{proof}
        Since the proof of the above theorem is rather technical and follows the steps of Theorem 1 form~\cite{KAMIL} we postpone it to the Appendix~\ref{App:SLTO}.
        \end{proof}

        Next we present how this upper bound can be saturated using ESLTO engine. 

        \begin{thm}\label{thm:ESLTO_properties}
            Consider an ESLTO engine with $d$-level subsystems in two different temperatures. Then, in the limit of an infinite number of rounds, all sharp states can be achieved with arbitrary precision, except the one on the highest level, for which maximal population one can obtain  is $p_{d,d} = e^{-\beta_h(E_d - E_{d-1})}$. Furthermore, one can achieve a state with that, or smaller, population on the maximally excited level and all remaining population arbitrary distributed on any other energy level.
        \end{thm}

        \begin{proof}
        We prove the theorem by constructing a protocol to obtain extremal (sharp) states using multiple strokes of ESLTO engine.
        
        Let us first consider a 2-dimensional subspace spanned by energy levels $(E_d,E_{d-1})$ and $(E_{d-1},E_d)$. The elementary thermal operation in this subspace--thermal swap--is given by 
        \begin{equation}
        S_{d,d-1}^{(\beta_h,\beta_c)} =  \left(\begin{matrix}
        1-g & 1\\
        g & 0
        \end{matrix}\right)\, ,
        \end{equation}
        where $g = e^{- (\beta_c - \beta_h)(E_{d} - E_{d-1})}$ is the ratio of populations between these two levels for the product of Gibbs states. 
        If one performs this swap and then repeats it after the temperature change, which results in changing the order of levels, the resulting operation is
        \begin{equation*}
        S_{d,d-1}^{(\beta_c,\beta_{h})} S_{d,d-1}^{(\beta_h,\beta_c)} = 
        \left(\begin{matrix}
        0 & g\\
        1 & 1-g
        \end{matrix}\right)
        \left(\begin{matrix}
        1-g & 1\\
        g & 0
        \end{matrix}\right) = 
        \left(\begin{matrix}
        g^2 & 0\\
        1-g^2 & 1
        \end{matrix}\right)\,.
        \end{equation*}
        Thus, by repeating those thermal swaps, one can move the population from energy levels $(E_d, E_{d-1})$, $(E_{d-1},E_d)$ to only one of those energy levels with exponential precision $g^{k}$ where $k$ is the number of strokes.
    
        Furthermore, after ``clearing'' one of those levels up to $\epsilon$, ex. $(E_d, E_{d-1})$, one can couple it with any other level $(E_n,E_m)$, with arbitrary $n$, $m$  except level $(E_d, E_d)$, using thermal swap. Thus taking $p_{n,m} e^{-\beta_h(E_d - E_n)} e^{-\beta_c(E_{d-1} - E_m)}$ or $p_{n,m} e^{-\beta_c(E_d - E_n)} e^{-\beta_h(E_{d-1} - E_m)}$ population from new level, depending on bath's attachment, and adding $\epsilon$ population from $(E_d,E_{d-1})$. 
        In the special case of coupling with a maximally excited state, 
        the thermal swap drops all population from it, while exciting only $\epsilon\; e^{-\beta_c (E_d - E_{d-1})}$.
    
        Thus, one can easily notice that after many repetitions of such "pumping" of 
        $\{(E_d,E_{d-1}), (E_{d-1}, E_d)\}$ subspace one can accumulate all populations except for exponentially decaying part $\delta$. Finally, using the compositions of thermal swap in this subspace, one can (up to exponential precision) create a sharp state on the level  $(E_d, E_{d-1})$.
        
        Technically, the operations discussed so far could be considered as a tree state in which levels $(E_{d-1},E_d)$ and $(E_{d},E_{d-1})$ are coupled and any other level is coupled to the one of those with zero population. However, there is a substantial difference between this scenario and the previous discussion of tree states, since at the temperature change the states $(E_{d-1},E_d)$ and $(E_{d},E_{d-1})$ exchange role between the ground one and the excited one. This unexpected phenomenon enables us to group all populations in a single energy level with almost maximal energy.
    
        This sharp state can then be transformed into any other sharp state (except of the maximally excited one) by using single thermal swap between the level with energies $(E_d, E_{d-1})$ and the level of interest.
        Finally, to obtain the maximal population on the highest energy level, one performs the thermal swap between $(E_d, E_{d-1})$, $(E_d,E_d)$ resulting in the state $(0,0,\cdots,1-e^{-\beta_h(E_d - E_{d-1})}, e^{-\beta_h(E_d - E_{d-1})})$, which ends the proof.
        \end{proof}

    For the sake of completeness, we point out that by the above proof in a 2-qubit engine based on SLTO operations maximal achievable population in ground state is $p_{00} = 1$, in the excited state in $p_{11} = e^{- \beta_h}$ and the maximal entanglement is obtained for the state with $p_{10} = 1$ or $p_{01} = 1$ and corresponds to $\mathcal{N}_{max} = 1/2$, which can be seen in the Fig.~\ref{fig:SLTO_engine}. Thus, one can easily compute the advantages and efficiency analytically presented also therein.
    From the above proof, we can also obtain the faithful monotone of the SLTO engine theory.

    \begin{cor}\label{SLTO_e_mono}
    The monotone  $\mathcal{M}(\vb{p}) := \max\{p_{d,d}- e^{-\beta_h(E_d - E_{d-1})},0 \}$ is a faithful monotone for SLTO engine theory.
    \end{cor}
    
    \begin{proof}
    By the above Lemma~\ref{SLTO_upper_bound}, we know that $\mathcal{M}(\vb{p})$ decreases monotonically under SLTO action, and one can easily check that it is zero for the product of Gibbs states.
    Furthermore, the proof of the above theorem demonstrated that in the SLTO engine, populations can be assembled on any level, thus, using the convex composition of different protocols, freely manipulated. The only exception is population $p_{d,d}$, which can only increase up to $e^{-\beta_h(E_d - E_{d-1})}$ (corresponding to the value of monotone $\mathcal{M} = 0$) and if this population is any larger, it can only decrease, by the Lemma~\ref{SLTO_upper_bound}.
    Hence, given the state $\vb{p}$, one can create any other energy-incoherent state $\vb{q}$ with the only restriction given by the Lemma~\ref{SLTO_upper_bound}, which ends the proof.   
    \end{proof}

    Furtheremore, we notice that this monotone can be directly related to engine efficiency of SLTO engine.

    \begin{obs}\label{obs:SLTO_eng_eff}
        Let us define a monotone for semilocal thermal operations as $\mathcal{M}_{SLTO}(\vb{p}) := \max_{ij} p_{ij} e^{\beta_h E_i} e^{\beta_c E_j} - \mathcal{Z}^{-1}$, where $\mathcal{Z} = \sum_{ij} e^{-\beta_h E_i} e^{-\beta_c E_j}$ is a partition function. Then, engine efficiency $\mathfrak{Ef}(s;\qty{\mathcal{M}})$ based on this monotone is proportional to SLTO engine monotone from Corollary~\ref{SLTO_e_mono}.
    \end{obs}

    \begin{proof}
    Proof of this property is shown in Appendix~\ref{app:obs_proof}.
    \end{proof}

        \begin{figure}[H]
            \centering
            \includegraphics[width=1\linewidth]{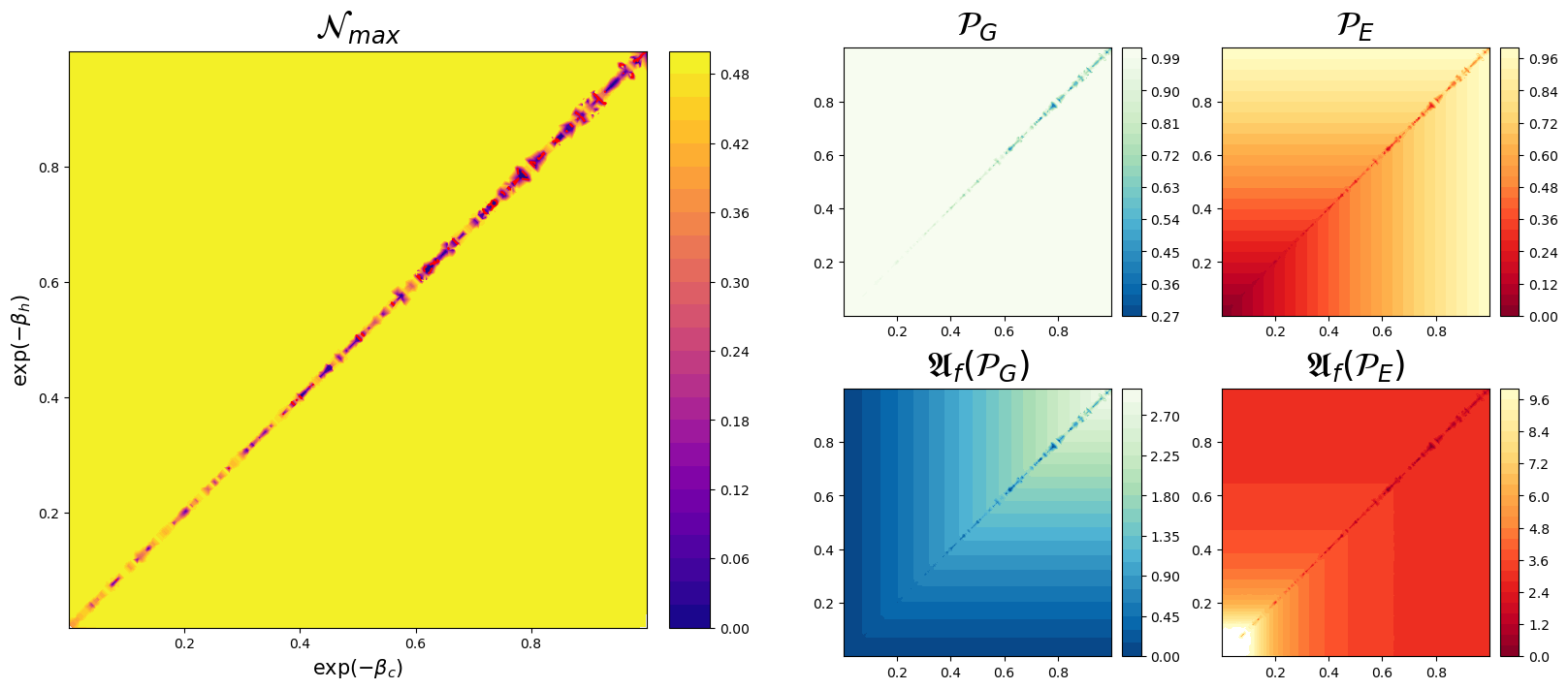}
            \caption{\textbf{SLTO engine:} Engine operating under semilocal thermal operations, with each stroke exchanging baths}, allows one to achieve maximal negativity $\mathcal{N}_{max}$ and ground state population $\mathcal{P}_G$ for arbitrarily small difference between $\beta_h$ and $\beta_c$, while the excited state is affected only by the higher temperature, making $\mathcal{P}_E = e^{-\beta_h}$. Note that the points near the diagonal are a result of the non-convergence of the numerical method. The excited state advantage $\mathfrak{A}_f(\mathcal{P}_E)$ has been clipped at $10$, as it goes to infinity in the limit of zero temperatures $\beta_h, \beta_c \rightarrow \infty$, since the excited state population in Gibbs states tends to zero faster than the bound from Theorem~\ref{thm:ESLTO_properties}. 
            \label{fig:SLTO_engine}
        \end{figure}

    Finally, we can extend the above results for the SLTO engine with a catalyst:

    \begin{cor}
    Consider an SLTO engine with (pair) of catalyst subspaces, then, in the limit of an infinite number of rounds, all states are achievable, even with qubit catalyst.
    \end{cor}

    \begin{proof}
    Consider a pair of qubit catalysts both in the ground state, $\vb{p}_{cat} = (1,0,0,0)$. Then, after combining the catalysts with an engine in the state $\vb{p}$ to obtain $\vb{p}_{cat}\otimes \vb{p}$, by the above theorem, one can transform them into arbitrary state of the form $\vb{p}_{cat}\otimes \vb{q}$, since it does not involve maximally excited state of a joint system.
    \end{proof}

    \subsection{LTOCC with memory advantage}

    As a final example we focus on an engine which is demonstrably beyond TO paradigm and is realised by supplying LTOCC with the memory of prior measurement results and the possibility of using them, referred to as LTOCC+M. To present the impact of memory, it is sufficient to consider parallel LTOCC and its symmetric variant, $\mathcal{S}$LTOCC, as described in~\cite{bistron2024local}.     
    We begin by discussing multiple rounds of parallel LTOCC operations at a single temperature followed by consideration of the engine, operating with the joint bath temperature control. 
    In this way, we can consider both symmetric and general parallel LTOCC in a single scenario.

    To describe parallel LTOCC operations, it is sufficient to focus on extremal channels.  Each channel consists of two thermal tensors, $M_{ij,kl} = T_{ikl}^{(1)} T_{jkl}^{(2)}$, each of them having two thermal operations as its layers. Since there are two extremal thermal operations for qubit, identity and the thermal swap, we have $2^4 = 16$ extremal operations in total. 
    Out of those, we focus on four channels:
    \begin{equation}
    \begin{aligned}
    & M^{(00)} =\left(
    \begin{array}{cccc}
     1 & 0 & 0 & 1 \\
     0 & 1-\gamma  & 0 & 0 \\
     0 & 0 & 1-\gamma  & 0 \\
     0 & \gamma  & \gamma  & 0 \\
    \end{array}
    \right),
    & M^{(01)} = \left(
    \begin{array}{cccc}
     1-\gamma  & 0 & 1-\gamma  & 0 \\
     0 & 1 & \gamma  & 0 \\
     \gamma  & 0 & 0 & 1 \\
     0 & 0 & 0 & 0 \\
    \end{array}
    \right), \\
    & M^{(10)} = \left(
    \begin{array}{cccc}
     1-\gamma  & 1-\gamma  & 0 & 0 \\
     \gamma  & 0 & 0 & 1 \\
     0 & \gamma  & 1 & 0 \\
     0 & 0 & 0 & 0 \\
    \end{array}
    \right),
    & M^{(11)} = \left(
    \begin{array}{cccc}
     (1-\gamma )^2 & 1 & 1 & 0 \\
     (1-\gamma ) \gamma  & 0 & 0 & 0 \\
     (1-\gamma ) \gamma  & 0 & 0 & 0 \\
     \gamma ^2 & 0 & 0 & 1 \\
    \end{array}
    \right),
    \end{aligned}
    \label{pLTOCC}
    \end{equation}
    each of which has a different sharp state as a fixed point. 
    Thus, one can obtain all two-qubit energy-incoherent states using convex combinations of multiple products of those LOCC operations.
    In thermal engines the multiple repetitions of the same channel form~\eqref{pLTOCC} also result in convergence to the appropriate sharp state.
    
    Note that the channel $M^{(11)}$ is also a symmetric LOCC which can be explicitly checked. Furthermore, one can use the remaining three extremal $\mathcal{S}$LTOCC, obtained from extremal bi-thermal tensors~\cite{bistron2024local},
    \begin{equation*}
    \begin{aligned}
    & M_S^{(1)} = \left(
    \begin{array}{cccc}
     (1-\gamma ) (1-(1-\gamma ) \gamma ) & 1-\gamma  & 1-\gamma  & 0 \\
     (1-\gamma )^2 \gamma  & \gamma  & \gamma  & 0 \\
     \gamma  (1-(1-\gamma ) \gamma ) & 0 & 0 & 1 \\
     (1-\gamma ) \gamma ^2 & 0 & 0 & 0 \\
    \end{array}
    \right), \\
    & M_S^{(2)} = \left(
    \begin{array}{cccc}
     (1-\gamma ) (1-(1-\gamma ) \gamma ) & 1-\gamma  & 1-\gamma  & 0 \\
     \gamma  (1-(1-\gamma ) \gamma ) & 0 & 0 & 1 \\
     (1-\gamma )^2 \gamma  & \gamma  & \gamma  & 0 \\
     (1-\gamma ) \gamma ^2 & 0 & 0 & 0 \\
    \end{array}
    \right), \\
    & M_S^{(3)} = \left(
    \begin{array}{cccc}
     (1-(1-\gamma ) \gamma )^2 & (1-\gamma )^2 & (1-\gamma )^2 & 1 \\
     (1-\gamma ) \gamma  (1-(1-\gamma ) \gamma ) & (1-\gamma ) \gamma  & (1-\gamma ) \gamma  & 0 \\
     (1-\gamma ) \gamma  (1-(1-\gamma ) \gamma ) & (1-\gamma ) \gamma  & (1-\gamma ) \gamma  & 0 \\
     (1-\gamma )^2 \gamma ^2 & \gamma ^2 & \gamma ^2 & 0 \\
    \end{array}
    \right),
    \end{aligned}
    \end{equation*}
    to change its sharp state into any other one.
    Once again we can extend this procedure to the $\mathcal{S}$LTOCC engine operating between two different temperatures by repeatedly using $M^{(11)}$ channel in both temperatures.

    Therefore, we conclude that all energy-incoherent states of two-qubit parallel LTOCC and $\mathcal{S}$LTOCC thermal engines are free and there exists an engine protocol that can transform one into any other (and so is the case for parallel LTOCC and $\mathcal{S}$LTOCC operations without restricted number of rounds in constant temperature).

\section{Conclusions}
\label{sec:conc}

In this work, we substantially extended the investigation of resource engines, as presented in~\cite{KAMIL}. Having in mind the potential use of resource engines utilising different resource theories, we start with an abstract formalisation of resource engines, which leads to general statements about the set of realisable operations and attainable states, with a necessary and sufficient condition for a resource engine to trivialise. In addition, we introduce quantifiers of the advantage provided by resource engines and identify a quantity, which we refer to as \textit{engine efficiency}, playing a role of resource monotone, which allows to consider engine theories as particular cases of resource theories.

Within this framework, we explore thermal engines based on different resource theories of athermality: thermal operations (TO), local thermal operations with classical communication (LTOCC), and semilocal thermal operations (SLTO)—-with a focus on their free states and operational limitations. To provide concrete insights, we examined key thermodynamic tasks, such as cooling, heating and state transformations, highlighting bounds on maximal excitation, de-excitation, and entanglement generation.

To provide analytically accessible interior bounds on the sets of free states of TO and LTOCC engines we devised a notion of \textit{tree-states}, which are a special class of engine's free states achievable (up to arbitrary precision) using simple graph-based protocols consisting solely of two-level operations.  Due to the physical realisability of elementary thermal operations via eg. Jaynes-Cummings interaction~\cite{jaynes-cummings_original, Shore1993}, tree-states give not only a theoretical approximation of the free states in a given thermal engine theory, but also a premise of experimental construction of such highly non-trivial out-of-equilibrium engine states. We also proved equivalence between different classes of LTOCC engines, reducing the need for protocols with more than two rounds per engine stroke. 

Our second achievement is a full description of free states in the engine for which working qubits can be coupled selectively to cold and hot baths, modelled by the SLTO. Additionally, we derived a faithful monotone, which turned out to be proportional to engine efficiency, and discussed advantages arising from the use of catalysts. We present an upper bound on the SLTO engine's free states, analogous to the one for the TO engine from~\cite{KAMIL}, followed by a constructive proof that this bound can be saturated using elementary SLTO operations.

In addition to theoretical results, we have conducted extensive numerical simulations with arbitrarily complex protocols for each of the engine models. In particular, a comparison between optimal states achieved numerically and analytical results obtained from tree-states shows that the interior bound obtained from the latter is not tight, thus requiring protocols beyond the graph-based structure of tree-states.

The outlook of our study is twofold. On one hand, the concept of resource engines, first presented in~\cite{KAMIL} and formalized in this work is of general purpose, thus may be compelling to consider resource engines of coherence~\cite{Streltsov2017}, magic~\cite{Veitch2014} or even a fully coherent treatment of resource theory of athermality based either on thermal operations or Gibbs-preserving operations~\cite{Korzekwa2017Arrow, Lostaglio2015coherentthermo, Gour2018qmajoriz, deoliveira2022geothermal, czartowski2024catalytic}. Furthermore, as pointed out in~\cite{KAMIL}, code-switching~\cite{Anderson2014, JochymOConnor2014}, used to enable universal gate set while applying error correction codes, can be reinterpreted through the lens of resource engines, suggesting potential optimizations for universal fault-tolerant quantum computation. Additionally, the stroke‑based model of resource engines considered here does not presently have a direct correspondence with the well‑explored framework of steady‑state dissipative engineering, particularly in the context of entanglement generation~\cite{BohrBrask2015, FrancescoDiotallevi2024, dols2025steady0state}. This reflects a fundamental difference between continuous Markovian dynamics, which typically lead to steady states independent of initial conditions, and discrete, potentially memory‑intensive engine protocols. Partial connections between these paradigms have begun to emerge through stroboscopic engines modelled by collisional frameworks~\cite{Molitor2020}, and recent formal links between Markovian thermal operations and energy‑preserving collisional models~\cite{Luo2026Equivalence} suggest possible routes toward extending the resource‑theoretic framework developed here to continuous‑time engines in future work.

On the other hand, the presented study of thermal engines is far from complete.
Although we conducted an intensive investigation of the set of free states, the structure of free operations and their associated monotones remains incomplete-—a challenge even for constant-temperature regimes~\cite{Mazurek_2018}. Future work could leverage asymptotic limits to gain deeper insights into engine operations, paving the way for both theoretical and experimental breakthroughs.

\section*{Acknowledgments}

The authors thank A. de Oliveira Jr. and Jeongrak Son for useful discussions and comments concerning the paper.
RB acknowledges support by the National Science Centre, Poland, under the contract number 2023/50/E/ST2/00472. JCz is supported by the start-up grant of the Nanyang Assistant Professorship at the Nanyang Technological University in Singapore, awarded to Nelly Ng.

\appendix

\section{Numerical approach to probing engine space}\label{app:numerical_proc}

Due to the exceeding number of extreme operations even for a single pair of hot-cold strokes, the exact analytical treatment for most thermal engine models we have presented turns out to be prohibitively complicated. Therefore, we have resorted to a numerical approach which we describe above.

Consider an engine operating under a selected pair of free operations $\mathcal{O}_1, \mathcal{O}_2 \subset \mathfrak{O}$ and free states $\mathcal{F}_1, \mathcal{F}_2 \subset \mathfrak{F}$ constituting a convex resource theory. To numerically obtain a set of engine's free states we applied the following steps:
\begin{enumerate}
    \item[0.] $S_0 = \qty{s_0}$ The starting state is taken as $s_0 \in \mathcal{F}_1$. 
    \item $S_{i+1} = \mathcal{O}_{(i+1\,\text{mod}\,2)+1} S_i = \qty{O s: O \in\mathcal{O}_{(i+1\,\text{mod}\,2)+1},\,s\in S_i}$: We use the operations from the free set for the second resource theory and apply it to all states in the current set of achievable states $S$.
    \item $S_{i+1}\mapsto \operatorname{ext}(\operatorname{conv}(S_{i+1}))$: Due to the convexity of the theory, we know that $\mathfrak{F}$ is a convex set, thus it is enough to keep the extremal points of convex hull of $S$ in each iteration.
    \item \textbf{Optional:} For any pair $\qty{s,s'}\subset S_{i+1}$ we remove $s'$ whenever $\qty[\frac{s}{\delta}] - \qty[\frac{s}{\delta}] = 0$ for $\delta\ll 1$ -- this step is employed to reduce computational load, eliminating states that are identical up to a selected rounding error $\delta$, while keeping the unrounded states as elements of $S$.
    \item Repeat steps 1-3 until either of the \textbf{break} conditions is met:
    \begin{enumerate}
        \item $\frac{\operatorname{vol}(S_{i+1}) - \operatorname{vol}(S_i)}{\operatorname{vol}(S_i)} \leq \epsilon$: Relative improvement in volume of achievable states reaches a threshold, signifying that $S_{i+1}\subset\mathfrak{F}$ provides a sufficiently good interior approximation of $\mathfrak{F}$.
        \item $i\geq N$: To guarantee end of the algorithm, an upper bound on the number of steps is implemented.
    \end{enumerate}
\end{enumerate}
For all numerical simulations, the parameters have been set to $N = 200,\,\epsilon = 10^{-7},\,\delta = 10^{-5}$ and temperatures have been randomly selected so that the distribution of $e^{-\beta_c}$ is uniform on $[0,1]$ interval. Additionally, for full TO and SLTO, step 1 could be simplified by generating $S_{i+1}$ by use of thermomajorization, which reduces the number of points generated while ensuring that the extreme points are included~\cite{Horodecki2013}.

\section{Extension of tree states and other bounds of thermal engines performance}\label{App:h_tree_&_Kamil}

    We start this Appendix by presenting lower bounds for achievable populations and entropy in the TO engine based on work~\cite{KAMIL}.
    To obtain the lower bounds on the maximal population in $00$ and $11$ levels, we can directly use Proposition 4 of~\cite{KAMIL}, which gives
    \begin{small}
    \begin{equation*}
    \begin{aligned} 
    \max p_{00} > & \frac{e^{-\text{$\Delta $E} \beta_c} \left(-e^{\text{$\Delta $E} \beta_c}+e^{\text{$\Delta $E} \left(2 \beta_c+\beta_h\right)}+e^{\text{$\Delta $E} \beta_h}-1\right){}^2}{\left(-2 e^{\text{$\Delta $E} \beta_c}+e^{\text{$\Delta $E} \left(\beta_c+\beta_h\right)}+e^{\text{$\Delta $E} \left(2 \beta_c+\beta_h\right)}+e^{\text{$\Delta $E} \beta_h}-1\right){}^2} \times\\
    &\times \frac{e^{\text{$\Delta $E} \beta_c} \left(\text{csch}\left(\text{$\Delta $E} \beta_h\right) \sinh \left(\text{$\Delta $E} \left(\beta_c+\beta_h\right)\right)+2\right)-1}{\text{csch}\left(\text{$\Delta $E} \beta_h\right) \sinh \left(\text{$\Delta $E} \left(\beta_c+\beta_h\right)\right)+2}\,,\\
    \max p_{11} > & \frac{e^{-\Delta E \beta_h} \sinh \left(\Delta E \beta_c\right)}{\sinh \left(\Delta E \left(\beta_c+\beta_h\right)\right)+2 \sinh \left(\Delta E \beta_h\right)}\,. \\
    \end{aligned}
    \end{equation*}
    \end{small}
    
    To estimate maximal entanglement one has to compute it for all ``extremal'' points $\vb{f^b}$ from Proposition 5 of~\cite{KAMIL}. The construction of these states is based on first accumulating or decreasing the population of the highest level, then accumulating or decreasing the remaining population at the next to the highest level without involving the higher one, and so on until all populations are fixed. The bitstrings $\vb{b}$ in the superscript denote whether the population was accumulated or decreased in consecutive states.
    We found that the state $\vb{f}^{[0,1,0]}$ yields maximal entanglement of this class of states. However, we omit the explicit formula due to its length. 

    We found out, that the use of tree-states states gives stronger bounds than $\vb{f^b}$ states. 
    Moreover, we consider tree-states beneficial since they provide a clear prescription on how to construct the state using only \textit{elementary} operations.

    Finally, we present generalization of tree-states, by combining them with the above-discussed highly non-equilibrium states.
    To do so, we introduce hyper-edges on the subset of vertices, corresponding to energy levels which can freely interact by allowed thermal operation. We treat them as new couplings and associate them with (unnormalized) $\vb{f^b}$ states in the same way as we associated pairs of vertices with the states $\vb{\tilde{\gamma}}$ or $\vb{\tilde{\Gamma}}$.
    By the same tokens as in the construction of tree-states, it is sufficient to consider only spanning trees of such hypergraphs, which ensures that the corresponding hypertree-states are well-defined.
    
    To provide an example of a hypertree state, consider state $\vb{p}$ for which the ratios of populations $(p_{00}, p_{01}, p_{11})$ are the same as in the state $\vb{f}^{[1,0]}$, whereas the ratio between populations $(p_{00}, p_{10})$ is the same as in the state $\vb{\tilde{\gamma}}$. By construction, such a state would have populations strongly shifted from $(E_1,E_0)$ level and toward $(E_0,E_1)$ level compared to the Gibbs state.
    However, for 2-qubit engines in no scenario we found hypertree states beneficial in comparison to $\vb{f}^b$ states or tree states.

\section{Technical results for SLTO}\label{App:SLTO}

    In this Appendix, we provide technical proof of Lemma~\ref{SLTO_upper_bound}, which content we invoke here for completeness. Throughout the proof we leverage the property, that for any two energy incoherent states $\vb{p}$, $\vb{q}$ there exists SLTO operation $M$ mapping one into another $\vb{q} = M \vb{p}$ if and only if one state termomajorize the other with respect to the product of Gibbs states $\vb{p} \succ_{\gamma_1 \otimes \gamma_2} \vb{q}$~\cite{Bera2021}.

    \textbf{Lemma 9.}\textit{
        Consider an SLTO engine operating on a pair of $d$-level subsystems with two different temperatures $\beta_h < \beta_c$.
        Then for each state $\vb{q}$ achievable form $\vb{p}$ in a possibly infinite number of rounds, the following holds:
        \begin{equation}
            q_{d,d} \leq \mathcal{M}_0(\vb{p}) := \max\{p_{d,d}, e^{-\beta_h(E_d - E_{d-1})} \}\,.
        \end{equation}
        Thus, $ \mathcal{M}_0(\vb{p}) := \max\{p_{d,d}, e^{-\beta_h(E_d - E_{d-1})} \}$ monotonically decrease under SLTO engine action.}

    \begin{proof}
    Let us denote future thermal cones of the set $X$ under SLTO operations with temperatures corresponding to the product of Gibbs states $\gamma_1$, $\gamma_2$ as $\mathcal{T}_{\gamma_{12}}(X)$ and $\mathcal{T}_{\gamma_{21}}(X)$  ie.
    \begin{equation*}
    \begin{aligned}
    & \mathcal{T}_{\gamma_{12}}(X) = \text{conv}[\{\vb{q}|\;\exists\, \vb{p} \in X: \vb{p} \succ_{\gamma_1 \otimes \gamma_2} \vb{q}\}]~, \\
    & \mathcal{T}_{\gamma_{21}}(X) = \text{conv}[\{\vb{q}|\;\exists\, \vb{p} \in X: \vb{p} \succ_{\gamma_2 \otimes \gamma_1} \vb{q}\}]~. \\
    \end{aligned}
    \end{equation*}
    Next, for any energy-incoherent state $\vb{p}$, let us define its ``almost maximally excited'' versions in the following way
    \begin{equation}
    \bar{\vb{p}}^{(12)} = (0,0,\cdots,1-p_{d,d},0,p_{d,d} ) ~~,~~~ \bar{\vb{p}}^{(21)} = (0,0,\cdots,0,1-p_{d,d},p_{d,d} )\,,
    \end{equation}
    with population $p_{d,d}$ occupying level $(E_d,E_d)$, while the population $(1-p_{d,d})$ occupy levels $(E_{d-1},E_d)$ and $(E_d,E_{d-1})$ respectively. It is straightforward to see that $\bar{\vb{p}}^{(12)} \succ_{\gamma_1\otimes\gamma_2} \vb{p}$ and $\bar{\vb{p}}^{(21)} \succ_{\gamma_2\otimes\gamma_1} \vb{p}$, thus
    \begin{equation*}
    \vb{p} \succ_{\gamma_1\otimes\gamma_2} \vb{q} \implies \bar{\vb{p}}^{(12)} \succ_{\gamma_1\otimes\gamma_2} \vb{q} ~~\text{and} ~~ \vb{p} \succ_{\gamma_2\otimes\gamma_1} \vb{q} \implies \bar{\vb{p}}^{(21)} \succ_{\gamma_2\otimes\gamma_1} \vb{q}~,
    \end{equation*}
    from which we can deduce the following inclusions between future cones
    \begin{equation*}
    \mathcal{T}_{\gamma_{12}}(\mathcal{T}_{\gamma_{21}}(\vb{p})) \subset \mathcal{T}_{\gamma_{12}}(\mathcal{T}_{\gamma_{21}}(\bar{\vb{p}}^{(21)})) \subset \mathcal{T}_{\gamma_{12}}(\{\bar{\vb{q}}^{(12)}| \vb{q} \in \mathcal{T}_{\gamma_{21}}(\bar{\vb{p}}^{(21)}) \})~,
    \end{equation*}
    where we abused the notation by substituting single-element set ex. $\{\vb{p}\}$ with its element $\vb{p}$. 
    Moreover, by step-by-step repeating the argumentation presented in Appendix C of~\cite{KAMIL} one can show that
    \begin{equation}
    \begin{aligned}
    & \max_{\vb{q}} \{q_{d,d}| \vb{q} \in \mathcal{T}_{\gamma_{xy}}(\bar{\vb{r}}^{(xy)})\} = \left\{
    \begin{matrix}
    r_{d,d} & \text{ for } r_{dd} \geq  \frac{1}{1 + e^{\beta_h(E_d - E_{d-1}) }}\\ 
    (1 - r_{d,d})e^{-\beta_h(E_{d} - E_{d-1}) } &\text{ for } r_{dd} \leq  \frac{1}{1 + e^{\beta_h(E_d - E_{d-1}) }}
    \end{matrix}
    \right. \\
    & \min_{\vb{q}} \{q_{d,d}| \vb{q} \in \mathcal{T}_{\gamma_{xy}}(\bar{\vb{r}}^{(xy)})\} = 0~,\\
    \end{aligned}
    \end{equation}
    where the indices $xy$ denote $12$ or $21$.
    Using above properties, we can upper bound the final population $q_{d,d}$ of the highest energy level after two interactions with thermal baths as
    \begin{equation}
    \begin{aligned}
    & \max_{\vb{q}}~~\{q_{d,d}|\mathcal{T}_{\gamma_{12}}(\mathcal{T}_{\gamma_{21}}(\vb{p}))\} \leq \\
    &\leq \max_{\vb{q}}~~ \left\{q_{d,d}| \vb{q} \in \mathcal{T}_{\gamma_{12}}(\{\bar{\vb{r}}^{(12)}| \vb{r} \in \mathcal{T}_{\gamma_{21}}(\bar{\vb{p}}^{(21)}) \}) \right\} = \\
    &  = \max_{\vb{q}} \max_{\vb{r} \in \mathcal{T}_{\gamma_{21}}(\bar{\vb{p}}^{(21)}) } \{q_{d,d}| \vb{q} \in \mathcal{T}_{\gamma_{12}}(\bar{\vb{r}}^{(12)}) \} = \\
    & = \max_{\vb{r} \in \mathcal{T}_{\gamma_{21}}(\bar{\vb{p}}^{(21)}) } \max_{\vb{q}} \{q_{d,d}| \vb{q} \in \mathcal{T}_{\gamma_{12}}(\bar{\vb{r}}^{(12)}) \} \leq \\
    & \leq \max_{\vb{r} \in \mathcal{T}_{\gamma_{21}}(\bar{\vb{p}}^{(21)}) } \max\left\{r_{d,d}, \; (1 - r_{d,d}) (1 - r_{d,d})e^{-\beta_h(E_{d} - E_{d-1}) }   \right\} \leq  \\
    & \leq \max\left\{  \max_{\vb{r} \in \mathcal{T}_{\gamma_{21}}(\bar{\vb{p}}^{(21)}) } r_{d,d}, \; (1 - \min_{\vb{r} \in \mathcal{T}_{\gamma_{21}}(\bar{\vb{p}}^{(21)}) } r_{d,d})e^{-\beta_h(E_{d} - E_{d-1}) }   \right\} \leq \\
    &\leq  \max\left\{p_{d,d}, (1 - p_{d,d})e^{-\beta_h(E_{d} - E_{d-1})}, e^{-\beta_h(E_{d} - E_{d-1})}  \right\} = \max \left\{p_{d,d},  e^{-\beta_h(E_{d} - E_{d-1})}  \right\}\,.
    \end{aligned}
    \end{equation}

    Since after one sequential interaction the final population of the highest energy level is bounded by either the initial population of this level or by a constant, we conclude that the same bound holds after arbitrarily many repetitions
    (strokes).
    \end{proof}

    \section{Proof of Observation~\ref{obs:SLTO_eng_eff}}
    \label{app:obs_proof}

    In order to show Observation~\ref{obs:SLTO_eng_eff} first let us argue that $\mathcal{M}_{SLTO}(\vb{p})$ is well defined. Since $p_{ij} e^{\beta_h E_i} e^{\beta_c E_j} \propto p_{ij}/\gamma_{ij}$ the value of $\mathcal{M}_{SLTO}(\vb{p}) +\mathcal{Z}^{-1}$ corresponds to the gradient of steepest, hence first, line segment in $\vb{p}$ termomajorization curve. Thus, by thermomajorization condition from Theorem~\ref{STLO_thermo_m}, it cannot increase under SLTO operations. The offset $\mathcal{Z}^{-1}$ guarantees $\mathcal{M}_{SLTO}(\vb{\gamma}) = 0$.

    To obtain the closed formula for engine efficiency, let us first calculate the value of its its constituents from ~\eqref{eff_def} and~\eqref{eq:free_advantage}. To calculate free advantages~\eqref{eq:free_advantage}, notice that there is only one free state from the opposite resource theory, the Gibbs state with inversely applied temperatures $\tilde{\vb{\gamma}}$, thus it is the optimal free state $f_*$ form~\eqref{eq:free_advantage}.
    The value of $\mathcal{M}_{SLTO}(\tilde{\vb{\gamma}})$ yields
    \begin{equation}
        \mathcal{M}_{SLTO}(\tilde{\vb{\gamma}}) = \max_{ij} \frac{e^{-\beta_c E_i} e^{-\beta_h E_j}}{\mathcal{Z}} e^{\beta_h E_i} e^{\beta_c E_j} - \mathcal{Z}^{-1} = \max_{ij} \frac{e^{-E_i(\beta_c - \beta_h)}e^{E_j(\beta_c- \beta_h)} - 1}{\mathcal{Z}} = \frac{e^{(E_d - E_1)(\beta_c - \beta_h)} - 1}{\mathcal{Z}}.
    \end{equation}
    Furthermore, maximum over engine operations in~\eqref{eq:free_advantage} is saturated, according to Lemma~\ref{SLTO_upper_bound} and Theorem~\ref{thm:ESLTO_properties}, by operation which maximally accumulate population on highest excited level. Thus 
    \begin{equation*}
        \max_{\hat{O} \in \mathfrak{O} } \mathcal{M}_{SLTO}(\hat{O} \vb{p}) = \max\{ p_{d,d} e^{\beta_h E_d} e^{\beta_c E_d},\;e^{-\beta_h(E_d - E_{d-1})}e^{\beta_h E_d} e^{\beta_c E_d}  \} - \mathcal{Z}^{-1} = \max\{p_{d,d} e^{\beta_h E_d} e^{\beta_c E_d},\; e^{\beta_h E_{d-1}} e^{\beta_c E_d}  \}- \mathcal{Z}^{-1},
    \end{equation*}
    where the second term in maximum is larger only if $\vb{p}$ is engine's free state.

    Thus, using symmetry between bath attachments and the fact any free operation acting on free state gives free state, we may express engine efficiency as
    \begin{equation}
    \begin{aligned}
        \mathfrak{Ef}(\vb{p}) &= \max\left( \frac{\max_{\hat{O} \in \mathfrak{O} } \mathcal{M}_{SLTO}(\hat{O} \vb{p}) - \mathcal{M}_{SLTO}(\tilde{\vb{\gamma}})}{\max_{\hat{O} \in \mathfrak{O} } \mathcal{M}_{SLTO}(\hat{O} \vb{\gamma}) - \mathcal{M}_{SLTO}(\tilde{\vb{\gamma}})}  - 1,0\right) \\
        & = \max\left( \frac{\left(\max_{\hat{O} \in \mathfrak{O} } \mathcal{M}_{SLTO}(\hat{O} \vb{p})\right) - \left(\max_{\hat{O} \in \mathfrak{O} } \mathcal{M}_{SLTO}(\hat{O} \vb{\gamma})\right)}{\max_{\hat{O} \in \mathfrak{O} } \mathcal{M}_{SLTO}(\hat{O} \vb{\gamma}) - \mathcal{M}_{SLTO}(\tilde{\vb{\gamma}})}  ,0\right) \\
        & = \max\left( \frac{p_{d,d} e^{\beta_h E_d} e^{\beta_c E_d} - e^{\beta_h E_{d-1}} e^{\beta_c E_d}}{e^{\beta_h E_{d-1}} e^{\beta_c E_d} - \frac{e^{(E_d - E_1)(\beta_c - \beta_h)}}{\mathcal{Z}}},0\right) \\
        & = \frac{e^{(\beta_c + \beta_h) E_d}}{e^{\beta_h E_{d-1}} e^{\beta_c E_d} - \frac{e^{(E_d - E_1)(\beta_c - \beta_h)}}{\mathcal{Z}}} \max\left\{p_{d,d} - e^{-\beta_h(E_d - E_{d-1})},0 \right\} \\
        & = \frac{e^{(\beta_c + \beta_h) E_d}}{e^{\beta_h E_{d-1}} e^{\beta_c E_d} - \frac{e^{(E_d - E_1)(\beta_c - \beta_h)}}{\mathcal{Z}}} \mathcal{M}(\vb{p})~,
    \end{aligned}   
    \end{equation}
    which ends the proof.

    \section{Additional figures for performance of thermal engines} \label{app:extra_figures}

    The following three figures follow the overall scheme for all the previous plots of negativity, ground state population and excited state populations, with data obtained from numerical procedure as described in Appendix~\ref{app:numerical_proc} for 2-round LTOCC with equal temperatures at each stroke, asymmetric temperatures at each stroke and for elementary thermal operations. We include them for completeness, but due to their qualitative similarity, they have been deferred. For better understanding of differences between them we refer the reader to Fig.~\ref{fig:tree_LTOCC_ETO_comparison}.

        \begin{figure}[H]
            \centering
            \includegraphics[width=1\linewidth]{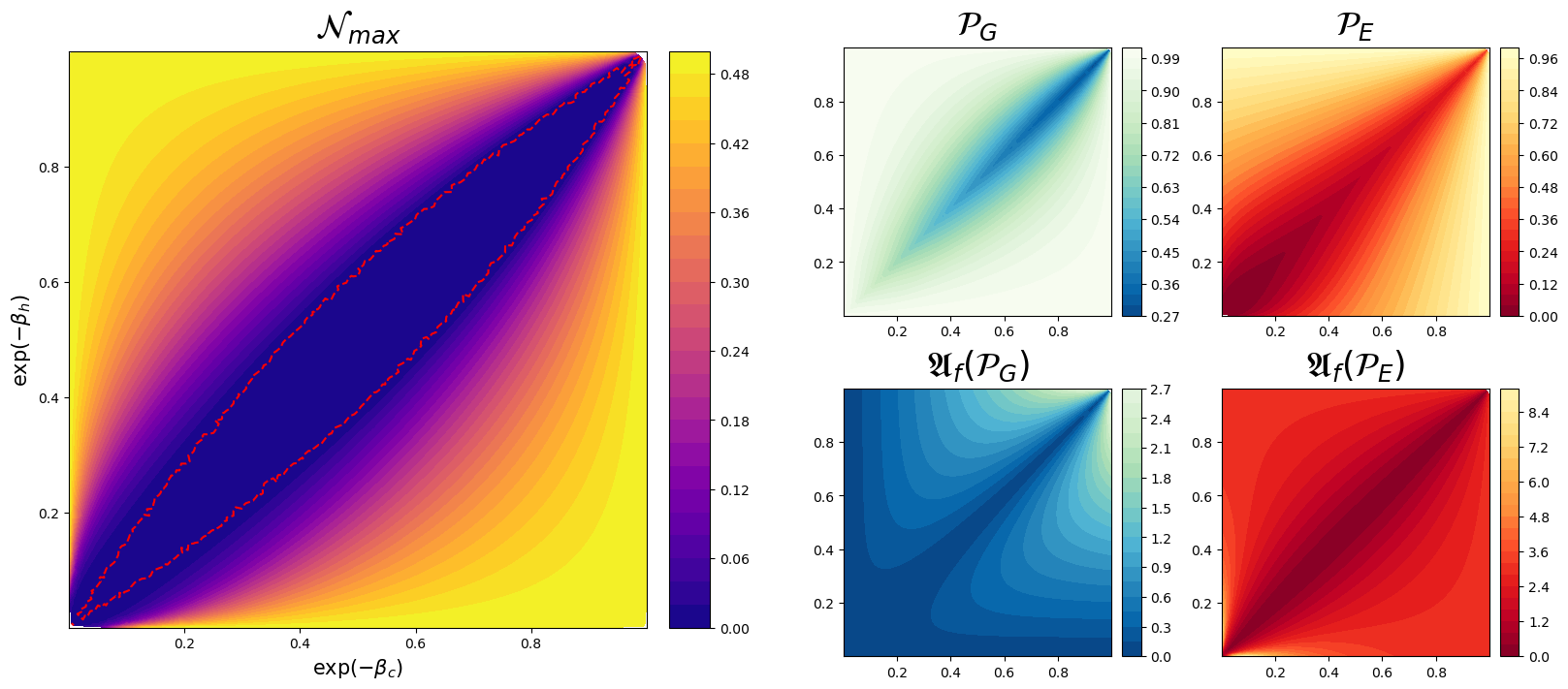}
            \caption{\textbf{Two-round LTOCC engine}}
            \label{fig:LTOCC_v1_2r_vals}
        \end{figure}

        \begin{figure}[H]
            \centering
            \includegraphics[width=1\linewidth]{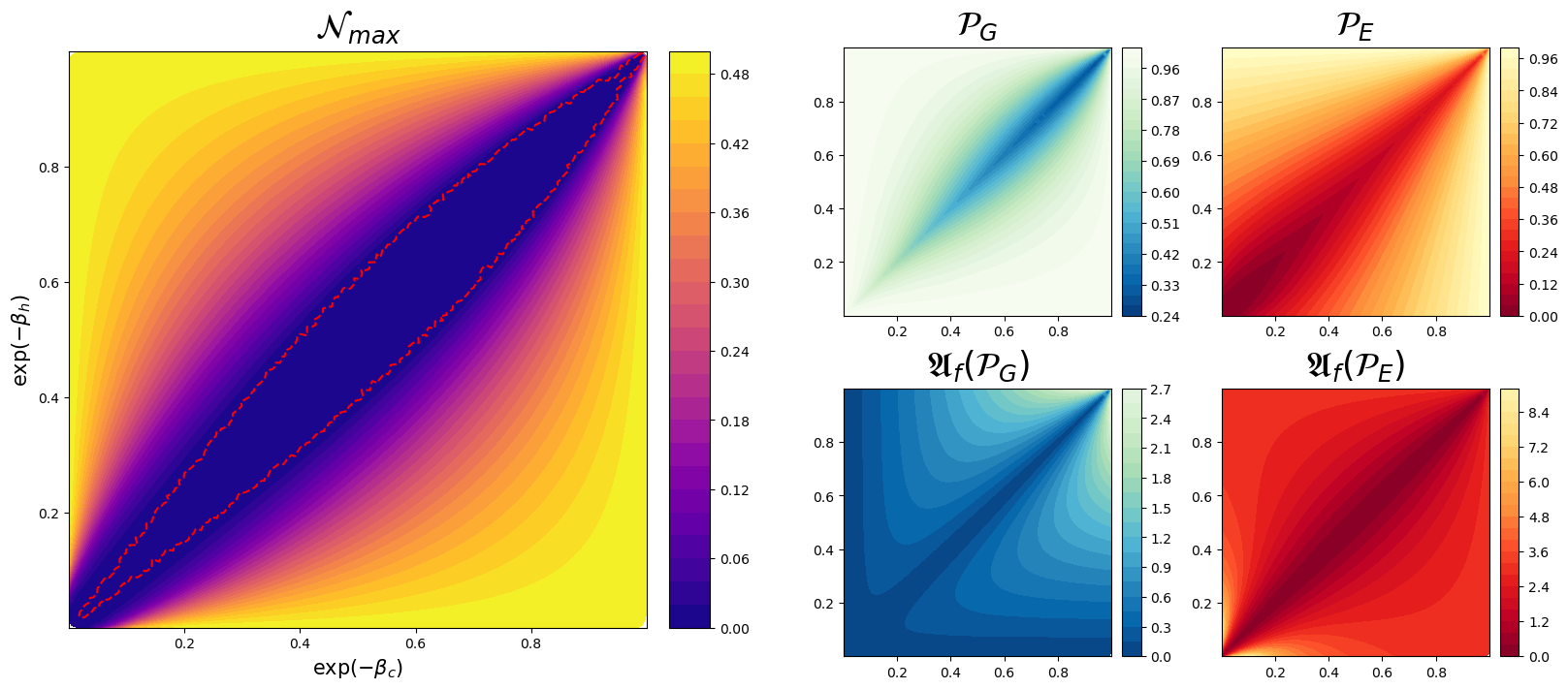}
            \caption{\textbf{ETO engine}}
            \label{fig:ETO_vals}
        \end{figure}

        \begin{figure}[H]
            \centering
            \includegraphics[width=1\linewidth]{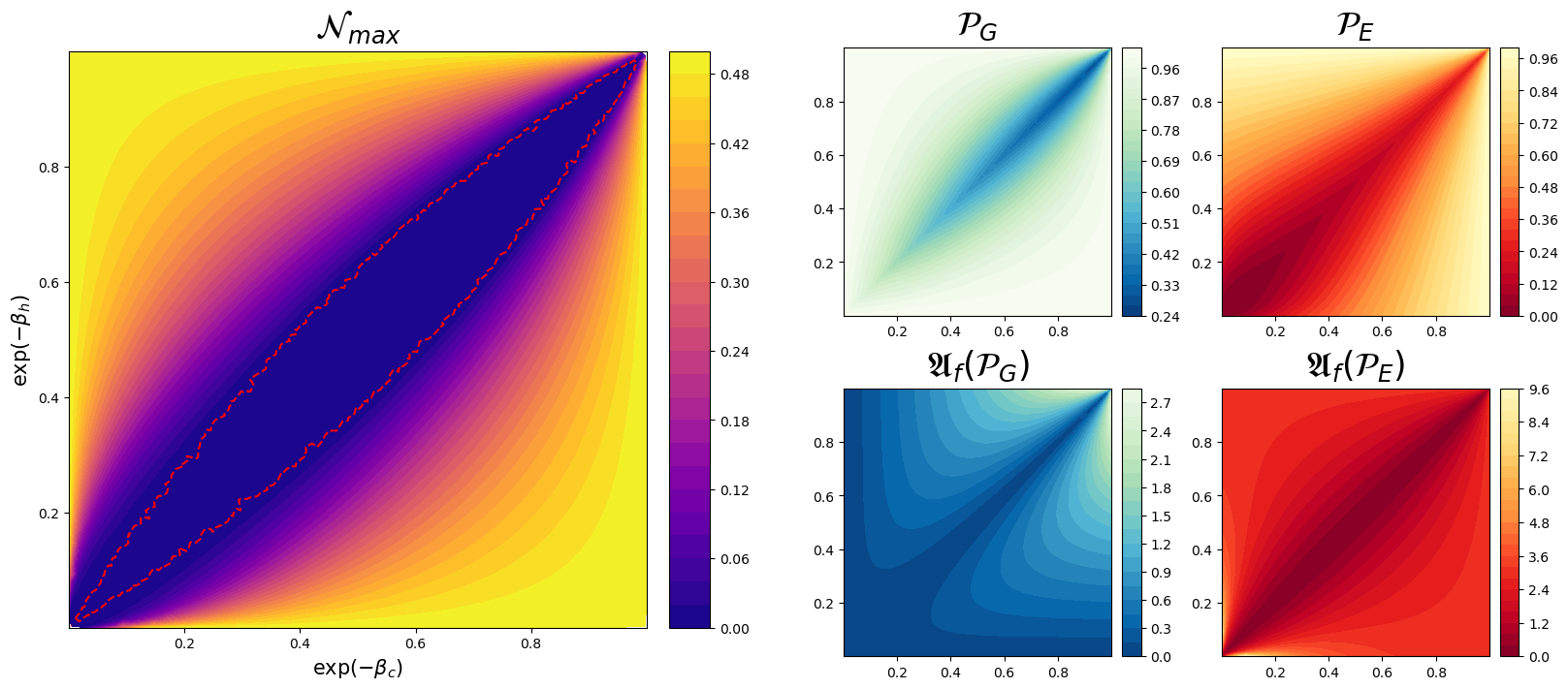}
            \caption{\textbf{Two-round LTOCC engine with asymmetric temperatures}}
            \label{fig:enter-LTOCC_2r_engine}
        \end{figure}

        Additionally, we include a figure depicting numerical results for an engine based on elementary SLTO. Note that the granular structure occurring around $\beta_c = \beta_h$ should be regarded as an artefact of numerical approach.
        
        \begin{figure}[H]
            \centering
            \includegraphics[width=1\linewidth]{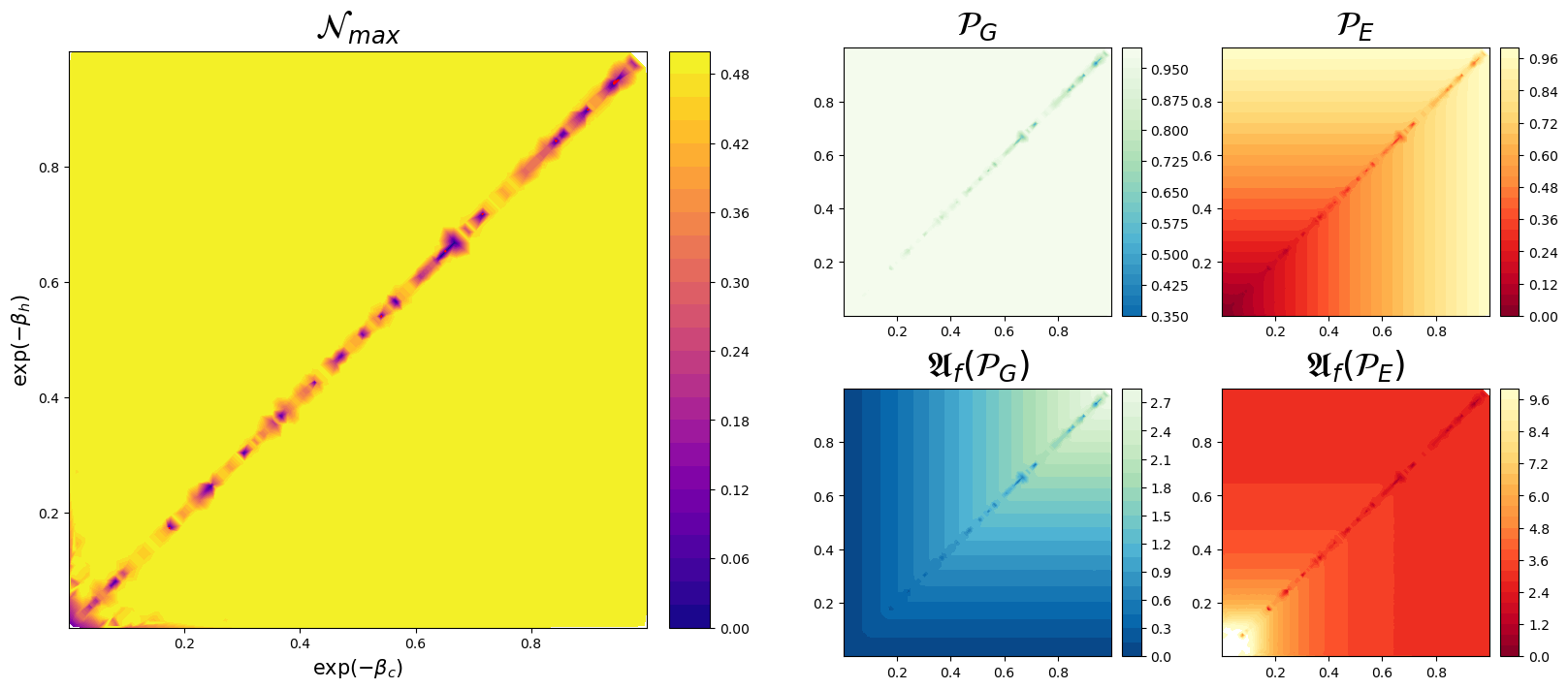}
            \caption{\textbf{ESLTO engine}}
            \label{fig:enter-ESLTO_engine}
        \end{figure}

        \newpage

        \section{Edge detection methods}\label{app:edge_detect}

        In this Appendix, we present a numerical study of critical behaviour in thermal engines. We start by discussing the methods and presenting the results for TO engine, for which we predicted critical temperatures~\eqref{eq:crit_temps}. Next, to validate our methods, we apply them to ETO tree-states for which the critical behaviour is fully analytically derived (Fig.~\ref{fig:treestates_reg}).

        \subsection{Edge detection for full thermal operations}

        Using the numerical results for engine based on full thermal operations, we have analysed the obtained data for edges, which are related to transitions in the behaviour of the respective quantities. The results of such an analysis are collected in Figure~\ref{fig:edge_detect}. Specifically, the data generated according to the procedure described in Appendix~\ref{app:numerical_proc} has been interpolated using linear and cubic interpolation to generate a $10^3\times10^3$ regular grid of points. Subsequently, the points have been subject to the double application of the Sobel edge detection procedure. Thus, for a given two-variable function $f$ we have approximated $\norm{\grad\qty(\norm{\grad f})}$, which has allowed us to find discontinuities of a second derivative by visual means. 
        \begin{figure}[H]
            \centering
            \includegraphics[width=.8\linewidth]{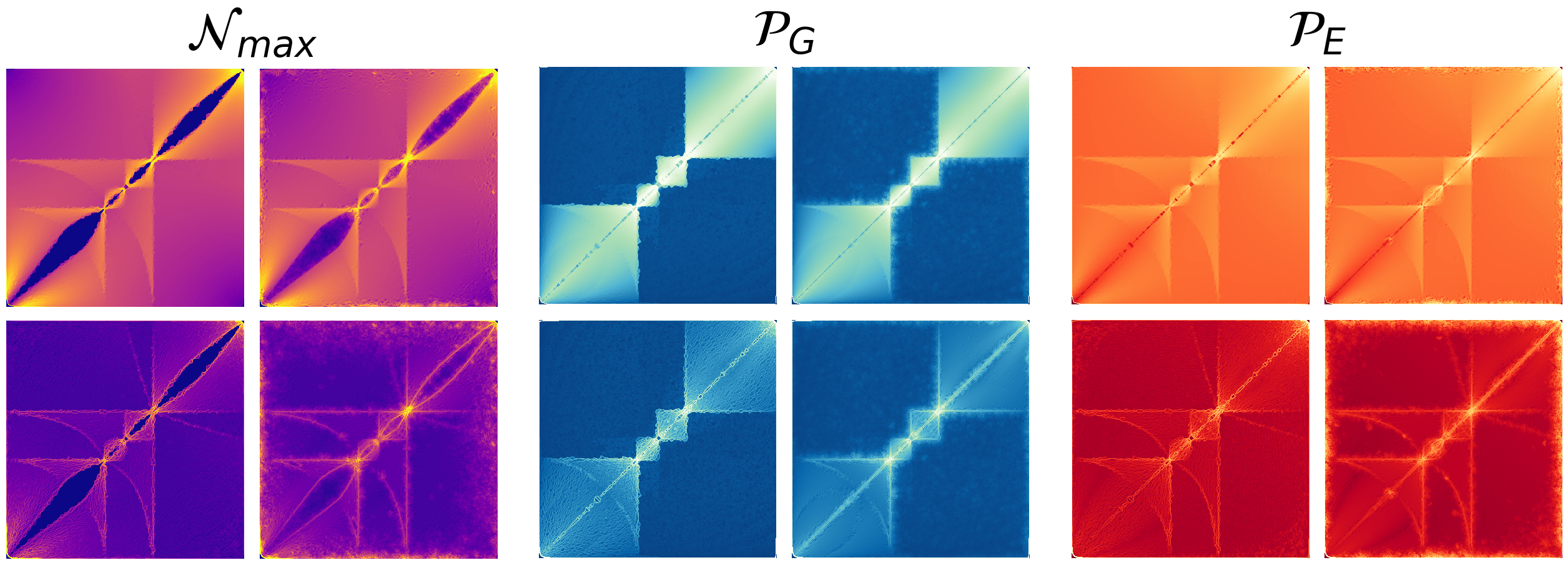}
            \caption{\textbf{Edge detection for full TO engine:} In order to identify all criticality behaviour akin to second-order phase transitions. To do this, we have employed the Sobel filter twice, with top row showing estimation of $\norm{\grad f}$ after first application and bottom row the estimation of $\norm{\grad\qty(\norm{\grad f})}$. In each subgrid left and right columns correspond to linear and cubic interpolation, respectively.}
            \label{fig:edge_detect}
        \end{figure}

        \begin{figure}[H]
            \centering
            \includegraphics[width=0.35\linewidth]{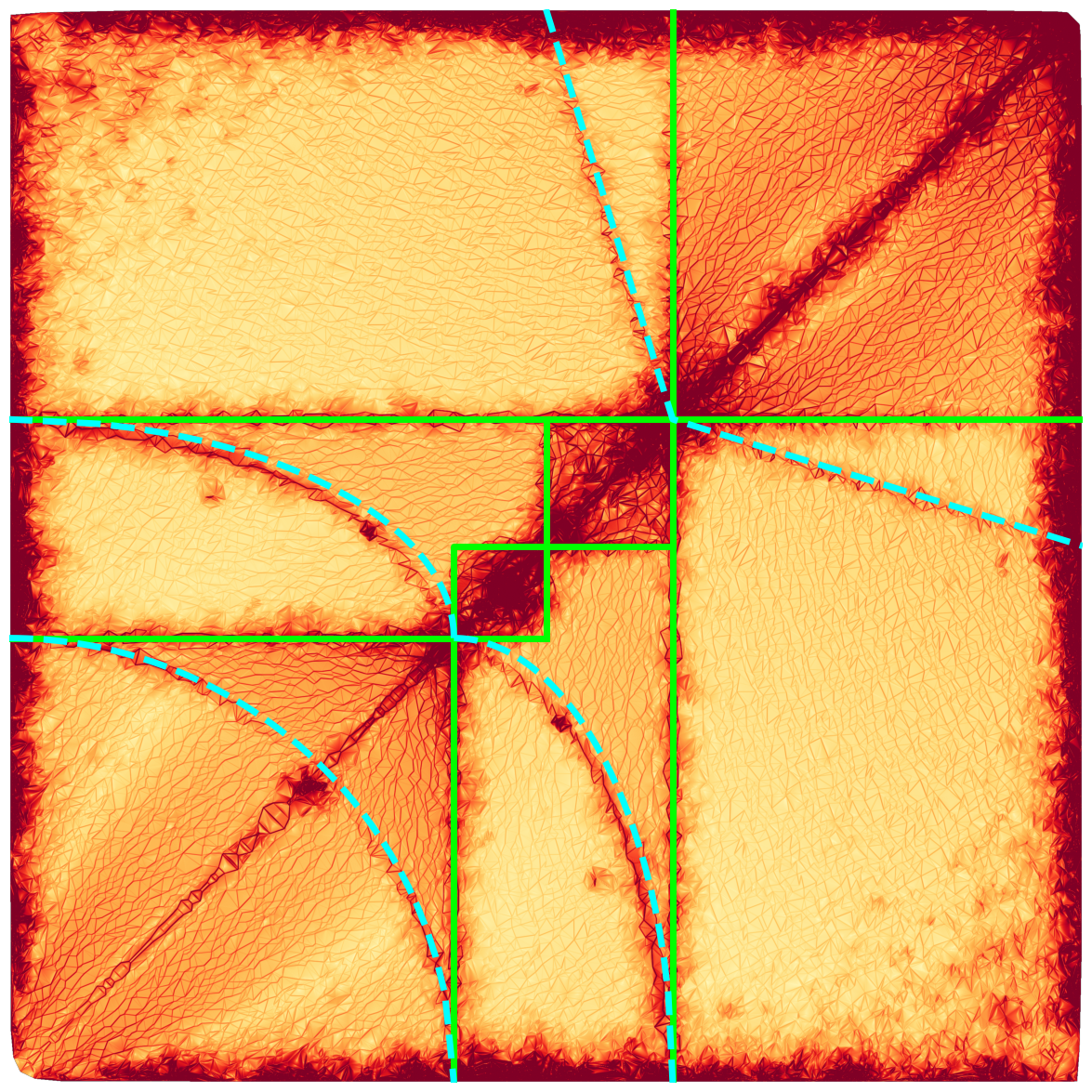}
            \caption{\textbf{Edge detection for $\mathcal{P}_E$ with full TO engine:} Using Sobel edge detection on average between cubic and linear interpolation, we obtain a better image of critical behaviour of $\mathcal{P_E}$ under full thermal operations' engine. Green lines correspond to critical temperatures that can be obtained directly from the properties of Gibbs state, while cyan dashed lines are approximation of behaviour observed in numerics that we have not been able to explain analytically.}
            \label{fig:ent_edge}
        \end{figure}

        In Fig.~\ref{fig:ent_edge} we focus on the edges detected in the $\mathcal{P}_E$ plot by taking an average of cubic and linear interpolations, thus obtaining a clearer picture of the critical lines. Note that all continuous green lines correspond to the already known critical values of $\beta$, namely $e^{-\beta}\in\qty{\sqrt{2}-1,1/2,(\sqrt{5}-1)/2}$. On the other hand, the cyan dashed lines are only approximate and except for the straight line connecting the points $((\sqrt{5}-1)/2,(\sqrt{5}-1)/2)$ and $(1/2,1)$ are just approximation and should not be taken as exact. Unfortunately, we have not been able to determine their origin.

        \subsection{Edge detection for tree-states}

        As an additional confirmation, we have applied the Sobel edge detection to the data that can be generated from analytical formulas for tree-states, which have been used to generate Figure~\ref{fig:treestates_vals}, as the boundaries between different optimal states have not been used explicitly in the generation method. Below we present the results, which match the boundaries as presented in Fig.~\ref{fig:treestates_reg}. 
        \begin{figure}[H]
            \centering
            \includegraphics[width=.9\linewidth]{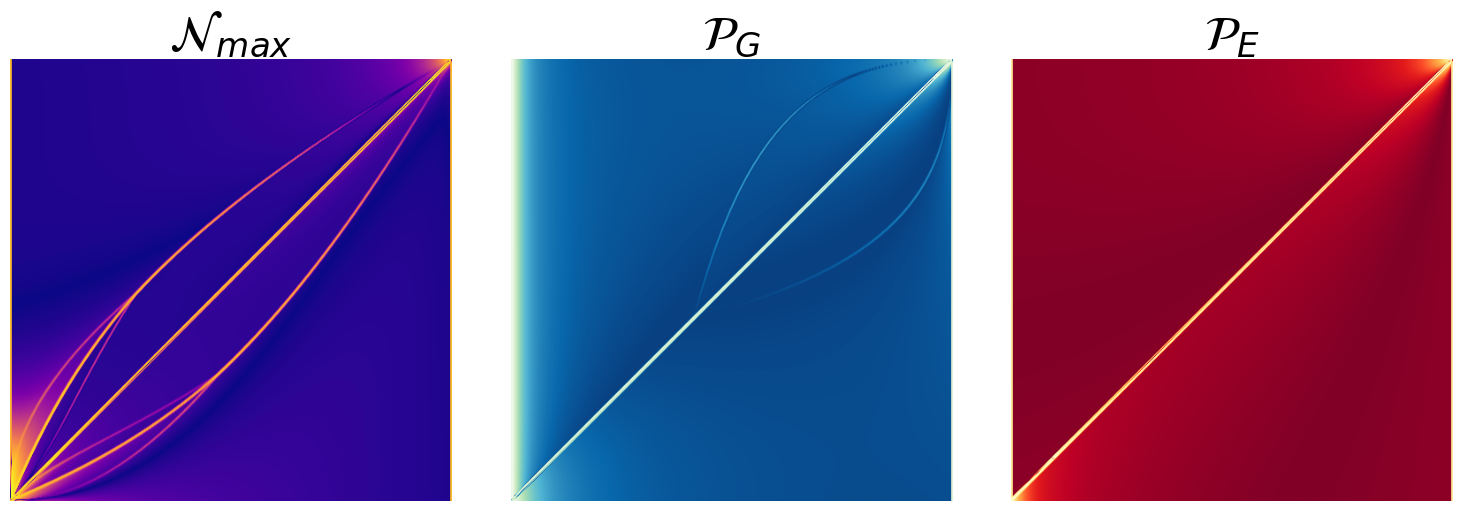}
            \caption{\textbf{Numerical retrieval of tree-state regions:} Division of regions obtained by numerical evaluation of tree-states' regions by Sobel edge detection should be compared against Figure~\ref{fig:treestates_reg} and equations~\eqref{tree_ltocc_heating_cooling} and~\eqref{tree_states_ETO}. Note that the plots do not exhibit complete uniformity, as the values change depending on the choice of specific temperatures.}
            \label{fig:treestate_edges}
        \end{figure}

\bibliographystyle{quantum_abbr}
\bibliography{references}

\end{document}